\documentclass[11pt,a4paper]{article}
\usepackage[utf8]{inputenc}
\usepackage{amsmath,amssymb,amsfonts,setspace,url,color}
\usepackage{fullpage}
\usepackage{latexsym}
\usepackage{amsthm}
\usepackage{subfig}
\usepackage{float}
\usepackage{algorithmic}
\usepackage[ruled]{algorithm}
\usepackage{bigstrut,array,multirow}
\theoremstyle{plain}
\newtheorem{theorem}{Theorem}
\newtheorem{lemma}{Lemma}

\newtheorem{corollary}{Corollary}

\newtheorem{definition}{Definition}

\newtheorem{proposition}{Proposition}

\theoremstyle{definition}

\newcommand{\bfs}{\mathsf{BFS}}
\newcommand{\dfs}{\mathsf{DFS}}
\newcommand{\ecc}{\mathrm{ecc}}
\newcommand{\opt}{\mathrm{opt}}
\newcommand{\leader}{\mathsf{leader}}
\newcommand{\init}{\mathsf{Initialization}}
\newcommand{\setup}{\mathsf{Setup}}
\newcommand{\checking}{\mathsf{Checking}}
\newcommand{\eval}{\mathsf{Evaluation}}
\newcommand{\eps}{\varepsilon}

\newcommand{\band}{\mathsf{bw}}
\newcommand{\nat}{\mathbb{N}}

\newcommand{\Gd}{\mathcal{G}_d}
\newcommand{\Ee}{\mathcal{E}}

\newcommand{\Ss}{\mathcal{S}}
\newcommand{\regR}{\mathsf{R}}

\newcommand{\regT}{\mathsf{T}}

\newcommand{\mem}{s}
\newcommand{\poly}{\textrm{poly}}
\newcommand{\DISJ}[1]{\textrm{DISJ}_{#1}}

\newcommand{\ceil}[1]{\left\lceil #1 \right\rceil}
\newcommand{\ket}[1]{| #1 \rangle}

\providecommand{\Aa}{\mathcal{A}}

\usepackage{tikz}
\usepackage{verbatim}
\usepackage{tkz-graph}
\usetikzlibrary{positioning,decorations.pathreplacing}
\newcommand{\Pp}{\mathcal{P}}

\begin{document}
\title{Sublinear-Time Quantum Computation of the Diameter in CONGEST Networks}
\author{
Fran{\c c}ois Le Gall\\
Graduate School of Informatics\\
Kyoto University\\
\url{legall@i.kyoto-u.ac.jp}\and
Fr\'ed\'eric Magniez\\
IRIF, Univ Paris Diderot, CNRS\\
\url{frederic.magniez@irif.fr}
}
\date{}

\maketitle
\thispagestyle{empty}
\setcounter{page}{1}
\begin{abstract}
The computation of the diameter is one of the most central problems in distributed computation. In the standard CONGEST model, in which two adjacent nodes can exchange $O(\log n)$ bits per round (here $n$ denotes the number of nodes of the network), it is known that exact computation of the diameter requires $\tilde \Omega(n)$ rounds, even in networks with constant diameter. In this paper we investigate quantum distributed algorithms for this problem in the quantum CONGEST model, where two adjacent nodes can exchange $O(\log n)$ \emph{quantum} bits per round. Our main result is a $\tilde O(\sqrt{nD})$-round quantum distributed algorithm for exact diameter computation, where $D$ denotes the diameter. This shows a separation between the computational power of quantum and classical algorithms in the CONGEST model. We also show an unconditional  lower bound $\tilde \Omega(\sqrt{n})$ on the round complexity of any quantum algorithm computing the diameter, and furthermore show a tight lower bound $\tilde \Omega(\sqrt{nD})$ for any distributed quantum algorithm in which each node can use only $\poly(\log n)$ quantum bits of memory.
\end{abstract}
\section{Introduction}
{\bf Diameter computation in CONGEST networks.}
The computation of the diameter is one of the most fundamental problems in distributed computing. In CONGEST networks, in which communication between nodes occurs with round-based synchrony and each channel has only $O(\log n)$-bit bandwidth, where $n$ denotes the number of nodes of the network, the diameter can be computed in $O(n)$ rounds \cite{Holzer+PODC12,Peleg+ICALP12}. A matching lower bound $\tilde \Omega(n)$ has been first shown by Frischknecht et al.~\cite{Frischknecht+SODA12},\footnote{In this paper the notation $\tilde O(\cdot)$ suppresses $\poly(\log n)$ factors, and the notation $\tilde \Omega(\cdot)$ suppresses $\frac{1}{\poly(\log n)}$ factors.} which also holds for sparse networks \cite{Abboud+DISC16} and even for deciding whether the network has diameter 2 or diameter 3 \cite{Holzer+PODC12}. The latter result immediately implies (see also \cite{Abboud+DISC16,Bringmann+DISC17}) that any distributed algorithm that computes a $(3/2-\varepsilon)$-approximation of the diameter, for any constant $\varepsilon>0$, requires $\tilde \Omega(n)$ rounds.

If larger approximation ratios are allowed, however, sublinear-time\footnote{As usual when discussing algorithms in the CONGEST model, the \emph{time complexity} of an algorithm refers to its round complexity. A sublinear-time algorithm means a $O(n^{1-\delta})$-round algorithm for some constant $\delta>0$.} approximation algorithms can be constructed. First note that a $2$-approximation of the diameter can trivially be computed in $O(D)$ rounds, where $D$ denotes the diameter of the graph, by computing the eccentricity of any node. Much more interestingly, Lenzen and Peleg constructed a $O(\sqrt{n} \log n+D)$-round $3/2$-approximation algorithm \cite{Lenzen+PODC13}, which was improved by Holzer et al.~to $O(\sqrt{n\log n}+D)$ rounds \cite{Holzer+DISC14}. \vspace{2mm}

\noindent
{\bf Quantum distributed computing.}
While it is well known that quantum communication can offer significant advantages over classical communication in several settings such as two-party communication complexity (see, e.g.,~\cite{deWolf02,Broadbent+08,Denchev+08}), there are relatively few results directly relevant to distributed network computation. One first evidence of the potential of quantum distributed computing was the design of exact quantum protocols for leader election \cite{Tani+12}. Gavoille et al.~\cite{Gavoille+DISC09} then considered quantum distributed computing in the LOCAL model, and showed that for several fundamental problems, allowing quantum communication does not lead to any significant advantage. The power of distributed network computation in the CONGEST model has recently been investigated by Elkin et al.~\cite{Elkin+PODC14}. In this model the nodes can use quantum processing and communicate using quantum bits (qubits): each edge of the network corresponds to a quantum channel (e.g., an optical fiber if qubits are implemented using photons) of bandwidth $O(\log n)$ qubits. The main conclusions reached in that paper were that for many fundamental problems in distributed computing, such as computing minimum spanning trees or minimum cuts, quantum communication does not, again, offer significant advantages over classical communication. The main technical contribution of \cite{Elkin+PODC14} was the introduction of techniques to prove lower bound for quantum distributed computation, which is significantly more challenging than proving lower bounds in the classical setting due to several specific properties of quantum information, such as quantum non-locality and the impossibility of ``copying'' quantum information. A pressing open question is to understand for which important problems in distributed computing quantum communication can help. \vspace{2mm}

\noindent
{\bf Our results.}
In this work we consider quantum distributed network computation in the CONGEST model, and especially investigate the complexity of computing the diameter. Our main contributions are for the exact computation of the diameter. We present the first quantum distributed algorithm that overcomes classical algorithms for this task:

\setlength{\extrarowheight}{1.5pt}
\begin{table}[tb]
 \begin{center}
  \begin{tabular}{|l|l|l|}
  \hline
   Problem & Classical & Quantum  \\ \hline\hline
   Exact computation&$O(n)$ \cite{Holzer+PODC12,Peleg+ICALP12}& $ O(\sqrt{nD})$ \:\hspace{13mm}Th.~\ref{th:UB1} \bigstrut\\ \hline
    \multirow{2}{*}{Exact computation}&\multirow{2}{*}{$\tilde \Omega(n)$ \cite{Frischknecht+SODA12}}& $\tilde \Omega(\sqrt{n}+D)$ \:\:\:\:\:\:\:\:\:\:\:Th.~\ref{th:LB1}\bigstrut \\ 
     &&$\tilde \Omega(\sqrt{nD/s}+D)$ \:\:Th.~\ref{th:LB2}\bigstrut\\ \hline
     \hline
   $3/2$-approximation &$ \tilde O(\sqrt{n}+D)$ \cite{Lenzen+PODC13,Holzer+DISC14}& $\tilde O(\sqrt[3]{nD}+D)$ \:\:\:\:\:\:\,Th.~\ref{th:UB2}\bigstrut\\ \hline 
   $(3/2-\varepsilon)$-approximation&$\tilde \Omega(n)$ \cite{Holzer+PODC12,Abboud+DISC16,Bringmann+DISC17}& $\tilde \Omega(\sqrt{n}+D)$ \:\:\:\:\:\:\:\:\:\:\:Th.~\ref{th:LB1}\bigstrut \\ \hline
  \end{tabular}
  \caption{Our results on the quantum round complexity of computing/approximating the diameter, and the corresponding known classical results. In this table $n$ denotes the number of nodes in the network, $D$ denotes the diameter and $\mem$ denotes the quantum memory used by each node.}
  \label{table:results}\vspace{-4mm}
 \end{center}
\end{table}
 
\begin{theorem}\label{th:UB1}
There exists a $\tilde O(\sqrt{nD})$-round quantum distributed algorithm, in which each node uses $O((\log n)^2)$ qubits of memory, that computes with probability at least $1-1/\poly(n)$ the diameter of the network, where $n$ denotes the number of nodes of the network and~$D$ denotes the diameter.
\end{theorem}
Theorem \ref{th:UB1} shows in particular that if the diameter is small (constant or at most polylogarithmic in~$n$), then it can be computed in $\tilde O(\sqrt{n})$ time in the quantum setting. Moreover, whenever the diameter is $O(n^{1-\delta})$ for some constant $\delta>0$, it can be computed in sublinear time. This significantly contrasts with the classical setting, where even distinguishing if the diameter is $2$ or $3$ requires $\tilde \Omega(n)$ time \cite{Holzer+PODC12}, as already mentioned. Note that our quantum algorithm uses only $O((\log n)^2)$ space, i.e., each node only needs to keep $O((\log n)^2)$ qubits of memory at any step of the computation. This low space complexity is especially appealing due to the technological challenges to construct large-scale quantum memory.  

A natural question is whether the upper bounds of Theorem \ref{th:UB1} can further be improved. We first show the following lower bound by adapting the argument from prior works that lead to the classical lower bound \cite{Holzer+PODC12}.
\begin{theorem}\label{th:LB1}
Any quantum distributed algorithm that decides, with probability at least $2/3$, whether the diameter of the network is at most~2 or at least 3 requires $\tilde \Omega(\sqrt{n})$ rounds.
\end{theorem}
Note that this result holds even for quantum algorithms with an arbitrary amount of memory. Theorem \ref{th:LB1} shows that the upper bound of Theorem \ref{th:UB1} is tight for networks with small diameter. A similar argument can actually be used to derive a $\tilde \Omega(\sqrt{n})$ lower bound for any quantum algorithm that decides if the diameter is at most $d$ or at least $d+1$ for larger values of~$d$, but for large diameters this does not match the upper bound of Theorem \ref{th:UB1}. We succeed in proving another lower bound, which matches the upper bound of Theorem \ref{th:UB1} for large diameters, under the assumption that the nodes in the distributed quantum algorithm use only small space. This is one of the main technical contributions of the paper.

\begin{theorem}\label{th:LB2}
Any quantum distributed algorithm, in which each node uses at most $s$ qubits of memory, computing with probability at least $2/3$ the diameter of the network requires $\tilde \Omega(\sqrt{nD/\mem})$ rounds.
\end{theorem}
Observe that Theorems \ref{th:UB1} and \ref{th:LB2} together completely (up to possible polylogarithmic factor) settle the complexity of distributed quantum exact computation of the diameter with small quantum memory. While the memory restriction in the lower bound is a significant assumption, we believe that Theorem \ref{th:LB2} is still a fairly general result since most known quantum communication protocols (e.g., protocols used to show the superiority of quantum communication in the model of communication complexity) use only a polylogarithmic amount of quantum memory. We actually conjecture that the upper bound of Theorem~\ref{th:UB1} is tight even without any restriction on the size of quantum memory used by the nodes.

We then consider approximation algorithms and show that the round complexity can be further decreased if we are only interested in computing a $3/2$-approximation of the diameter.
\begin{theorem}\label{th:UB2}
There exists a $\tilde O(\sqrt[3]{nD}+D)$-round quantum distributed algorithm
 that computes with probability at least $1-1/\poly(n)$ a $3/2$-approximation of the diameter of the network. 
\end{theorem}
Whenever the diameter is small, this quantum approximation algorithm again provides a significant improvement over the best known classical algorithms \cite{Lenzen+PODC13,Holzer+DISC14} already mentioned. 
The quantum algorithm of Theorem \ref{th:UB2} actually consists of two phases, one classical and one quantum. The quantum phase still has polylogarithmic memory per node.


Our results are summarized in Table \ref{table:results}. \vspace{2mm}

\noindent
{\bf Overview of our upper bound techniques.}
At a high level, our approach can be described as a distributed implementation of quantum search and its generalizations (quantum amplitude amplification and quantum optimization), which are fundamental quantum techniques well studied in the centralized setting and in two-party quantum communication. In Section \ref{prelim:ampl} we show how to implement them in the distributed setting by electing a leader in the network who will coordinate the quantum search. We develop a general framework involving three basic operations, $\init$, $\setup$ and $\eval$, which need to be implemented on the whole network to perform quantum distributed optimization. This framework is general and can be applied to a large class of search or optimization problems on a distributed network as long as quantum communication is allowed between the nodes. The round complexity of the whole approach depends on the round complexity of each of the three basic operations, which is naturally problem-specific. We show how to implement them efficiently for the case of diameter computation. 

For exact diameter computation (Theorem \ref{th:UB1}), the optimization problem we consider corresponds to finding a set $S$ of $\Theta(n/D)$ nodes that contains a vertex of maximum eccentricity in the network. The most delicate part is the implementation of the $\eval$ operation, which given any set $S$ distributed (as a quantum superposition) among all the nodes of the network needs to compute the maximum eccentricity among all nodes in $S$. Our implementation is based on a refinement of the deterministic classical distributed algorithm for shortest paths in~\cite{Peleg+ICALP12}. For our 3/2-approximation of the diameter (Theorem \ref{th:UB2}), the strategy is slightly different: our approach is inspired by the classical algorithm for $3/2$-approximation by Holzer et al.~\cite{Holzer+DISC14}. \vspace{2mm}

\noindent
{\bf Overview of our lower bound techniques.}
Essentially all the known lower bounds on the classical complexity of computing or approximating the diameter in the CONGEST model are obtained by an argument based on the two-party communication complexity of the disjointness function. More precisely, the idea is to first reduce the two-party computation of $\DISJ{k}$, the disjointness function on $k$-bit inputs, for some value of $k$, to the computation or the estimation of the diameter of a carefully constructed network, and then use the fact that the two-party classical communication complexity of $\DISJ{k}$ is $\Omega(k)$ bits \cite{Kalyanasundaram+92,Razborov92}. A first approach to obtain quantum lower bounds would be to use instead the (tight) quantum lower bound $\Omega(\sqrt{k})$ showed by Razborov \cite{Razborov03} for the quantum two-party communication complexity of $\DISJ{k}$. It turns out, however, that the dependence in $\sqrt{k}$ is too weak to lead to non-trivial lower bounds on the quantum round complexity of diameter computation. Fortunately, we observe that we can instead use a recent result by Braverman et al.~\cite{BGK+15}, which was obtained using quantum information complexity, showing that the $r$-message quantum two-party communication complexity of $\DISJ{k}$ is $\tilde \Omega(k/r +r)$. This better dependence in $k$ enables us to prove Theorem \ref{th:LB1}.
  
Proving the lower bound of Theorem \ref{th:LB2}, i.e., making the diameter appear in the lower bound, is much more challenging. Our approach is as follows. We start from the network introduced in \cite{Abboud+DISC16} to prove a $\tilde \Omega(n)$-round classical lower bound for deciding whether a sparse bipartite network has diameter~4 or diameter 5. We convert each edge between the left part and the right part of the network in this construction into a path of $d$ dummy nodes (the total number of introduced node is $\Theta(n)$), which immediately gives a reduction from the computation of $\DISJ{k}$ with $k\approx n$ into deciding whether the diameter is $d+4$ or $d+5$. Since $d$ rounds of communication are needed to transfer 1 bit (or 1 quantum bit) from the left part to the right part of the network, we would expect the following statement to be true: any $r$-round algorithm computing the diameter can be converted into a $O(r/d)$-message two-party protocol for $\DISJ{k}$ with $\tilde O(r)$ qubits of communication. This would give, via the lower bound from \cite{BGK+15} mentioned above, the claimed lower bound. The hard part is to prove this statement. While it is easy to see that the statement is true in the classical setting, in the quantum setting several difficulties arise (as in \cite{Elkin+PODC14}) due to entanglement between the nodes of the network. Our main technical contribution (Theorem \ref{th:sim} in Section \ref{sec:LB2}) shows that the statement is indeed true in the quantum setting if we assume that each node has small enough quantum memory. 

\section{Preliminaries}\label{sec:prelim}
We assume that the reader is familiar with the basic notions of quantum computation and refer to, e.g., \cite{Nielsen+00} for a good introduction.
A quantum state will be written using the ket notation such as~$\ket{\psi}$.
A quantum register means a set of quantum bits (qubits).
When the system has several identified registers, we indicate them as a subscript. For instance $\ket{i}_\mathsf{R}\ket{j}_\mathsf{Q}$, means that register $\mathsf{R}$ is in state $\ket{i}$, and register $\mathsf{Q}$ in state~$\ket{j}$. Of course those could be entangled and in superposition. Then we  write 
$\ket{\psi}=\sum_{i,j}\alpha_{ij}\ket{i}_\mathsf{R}\ket{j}_\mathsf{Q}$.

In this paper we will use the term \emph{CNOT copy} to refer to the unitary operation on $2m$ qubits (for some integer $m$) that maps the state $\ket{u}\ket{v}$ to the state $\ket{u}\ket{u\oplus v}$ for any binary strings $u,v\in\{0,1\}^m$, where $u\oplus v$ is the bitwise XOR operation on $u$ and $v$. This quantum operation can be implemented by applying $m$ Controlled NOT gates. In particular it maps the state $\ket{u}\ket{0}$, for any binary string $u$, to the state $\ket{u}\ket{u}$, and thus corresponds to a classical copy.

For any integer $M$, we denote by $[M]$ the set $\{1,2,\ldots,M\}$.
For an undirected graph $G=(V,E)$, 
we write $d(u,v)$ the distance between $u,v\in V$, i.e., the minimum length of a path between $u$ and $v$ in~$G$.
The \emph{eccentricity} of $u\in V$, denoted  $\ecc(u)$, is the maximum
distance of $u$ to any other node $v\in V$.
Then the \emph{diameter} $D$ is simply the maximum eccentricity among all nodes $u\in V$.
Finally, the Breath First Search tree from $u$ on $G$ is denoted by $\bfs(u)$
and the Depth First Search tree is denoted $\dfs(u)$.

\subsection{Quantum CONGEST networks}
In this paper we consider the CONGEST communication model. The graph $G = (V, E)$ represents the topology of the network, executions proceed with round-based synchrony and each node can transfer one message of $\band$ qubits to each adjacent node per round. Initially the nodes of the network do not share any entanglement. In this paper all the networks are undirected and unweighted. Unless explicitly mentioned, the bandwidth $\band$ will always be $\band=O(\log n)$, where $n=|V|$. The only exception is the results in Section~\ref{sec:sim}, and in particular~Theorem \ref{th:sim}, where the bandwidth will be kept as a parameter.
All links and nodes (corresponding to the edges and vertices of $G$, respectively) are reliable and suffer no faults. Each node has a distinct identifier.
Initially, each node knows nothing about the topology of the network except the set of edges incident to itself and the value $n$. 

We  write $\ket{\psi}_v$ to denote a state in the memory space of node $v$.
When such a state is entangled, we  use the tensor product notation.
For instance
$\frac{1}{\sqrt{M}}\sum_{i\in [M]} \bigotimes_{v\in V}\ket{i}_{v}$
denotes a uniform superposition over integers in $[M]$ distributed over all nodes of $V$ (i.e., each node 
is synchronized with the same value, in superposition).

When discussing the space complexity of quantum distributed algorithms, the memory refers to the number of qubits each node uses in its workspace. We assume for simplicity that each node performs unitary operations and measures its register only at the very end of the algorithm.\footnote{More general models, such as models allowing intermediate measurements and models that allow nodes to use both a quantum memory and a classical memory, can naturally be considered as well. We prefer not introducing such models in this paper since the precise definition (and in particular the relation between the classical part and the quantum part), which can for instance be proposed similarly to \cite{LeGall09}, would be fairly technical. We just mention that under natural ``hybrid'' models allowing both classical and quantum parts, the lower bound of Theorem \ref{th:LB2} would hold for any quantum algorithm that uses $s$ qubits of quantum memory, independently to the quantity of classical memory used, since the difficulty only comes from the quantum memory. A very simple example is the straightforward generalization of our model that allows a first classical phase with unbounded classical memory and then a quantum phase using only $s$ qubits of memory per node.} 

\subsection{Two-party communication complexity}\label{sec:prelim-qcc}
Let $X$, $Y$ and $Z$ be three finite sets. Consider two players, usually called Alice and Bob, and assume that Alice receives as input an element $x\in X$, while Bob receives an element $y\in Y$. In the model of communication complexity, first introduced in the classical two-party setting by Yao \cite{YaoSTOC79}, the players want to compute a function $f\colon X \times Y\to Z$ by running a protocol such that, at the end of the protocol, both Alice and Bob obtain $f(x,y)$, and they want to minimize the communication. In the quantum communication model, introduced by Yao \cite{YaoFOCS93}, the players are allowed to communicate with qubits. More precisely, the quantum communication complexity of a quantum protocol is the maximum (over all inputs) number of qubits that the protocol sends. The quantum communication complexity of $f$ is the minimum communication complexity of any quantum protocol that computes $f$ with probability at least 2/3.

For any integer $n\ge 1$, the disjointness function $\DISJ{k}\colon \{0,1\}^k\times\{0,1\}^k\to \{0,1\}$ is the function such that $\DISJ{k}(x,y)=0$ if and only if there exists an index $i\in\{1,\ldots,k\}$ such that $x_i=y_i=1$. It is well known that its (randomized) classical communication complexity is $\Theta(k)$ bits \cite{Kalyanasundaram+92,Razborov92}. Its quantum communication complexity is $\Theta(\sqrt{k})$ qubits \cite{Buhrman+STOC98,Hoyer+STACS02,Razborov03,Aaronson+05}. Recently Braverman et al.~\cite{BGK+15} proved the following lower bound for quantum protocol with limited interaction, i.e., when only a bounded number of messages can be exchanged between Alice and Bob, which significantly improved the previous bound from \cite{Jain+FOCS03}

\begin{theorem}[\cite{BGK+15}]\label{th:BGK+15}
The $r$-message quantum communication complexity of $\DISJ{k}$ is $\tilde \Omega(k/r+r)$. 
\end{theorem}
This result shows in particular that any $O(\sqrt{k})$-qubit quantum protocol for $\DISJ{k}$ requires $\tilde \Omega (\sqrt{k})$ messages, i.e., $\tilde \Omega (\sqrt{k})$ rounds of interaction between Alice and Bob.

\subsection{Quantum generic search}\label{app:dqo}
We now review the general framework of quantum generic search in the centralized model (see for instance~\cite{Magniez+SICOMP11} for a more thorough treatment). In this framework, we are given a set of items $X$ and we are looking for a marked item $x\in M$, for some unknown subset $M\subseteq X$. We have two main black-box unitary operators, and their inverses, available to design a quantum search procedure. The first one, $\setup$, is the quantum analogue of a random sampling. The second one, $\checking$, is a quantum checking procedure. 

It will be convenient to assume the existence of a third procedure, $\init$, that performs some global initialization: the procedure is applied only once at the beginning of the algorithm and its output is later used in order to implement $\setup$ and $\checking$.

We identify a specific register $\ket{\cdot}_I$ that we call \emph{internal}, which is the core of the algorithm. Its size is polylogarithmic in all parameters. It is used to encode the element we are looking for, to perform basic operations such as arithmetic operations (e.g., counting), and also to control the application of other unitary matrices. The quantum algorithm uses two additional registers to represent additional information created by the procedures $\init$ and $\setup$. 

The framework thus assumes that the following three quantum procedures are given (as black-boxes).
\begin{description}
\item[$\init$:] Creates an initial state $\ket{0}_I\ket{\text{init}}$ 
with some possible precomputed information $\ket{\text{init}}$ and a distinguished element $0\in X$, usually a long enough bit string of $0$s.
\item[$\setup{}$:] Produces a superposition from the initial state:
$$\ket{0}_I\ket{\text{init}}\mapsto \sum_{x\in X} \alpha_x\ket{x}_I\ket{\text{data}(x)}\ket{\text{init}},$$
where the $\alpha_x$'s are arbitrary amplitudes and $\text{data}(x)$ represents some information depending on $x$.
\item[$\checking{}$:] Performs the transformation
$$\ket{x,0}_I\ket{\text{data}(x)}\ket{\text{init}} \mapsto 
\ket{x,b_x}_I\ket{\text{data}(x)}\ket{\text{init}},$$
where $b_x=1$ if $x\in M$ and $b_x=0$ otherwise.
\end{description}\vspace{2mm}

Note that the procedures $\init$ and $\checking$ are often described as deterministic or randomized (i.e., classical) procedures. They can then been quantized using standard techniques: one first transforms it to a reversible map using standard techniques~\cite{Bennett+SICOMP89}, with potentially additional garbage 
whose size is of the same order as the initial memory space.

Define $P_M=\sum_{x\in M}|\alpha_x|^2$, the probability to observe a marked element when measuring Register~$I$ after one application of $\setup{}$ on the initial state.
Classically one could iterate $\Theta(1/P_M)$ time this process in order to get at least one marked element with high probability. Amplitude amplification explains how to get a marked element using simply $\Theta(1/\sqrt{P_M})$ iterations of $\setup{}$ and $\checking{}$.
\begin{theorem}[Amplitude amplification~\cite{Brassard+ICALP98}]\label{thm:aa}
Let $\eps>0$ and assume that either $P_M=0$ or $P_M\geq \eps$ holds.
Then, for any $\delta>0$, there is a quantum algorithm that can decide if $M=\emptyset$ with success probability at least  $1-\delta$
using one unitary operator $\init$,
{$O(\log(1/\delta)/\sqrt{\eps})$} unitary operators $\setup{}$ and $\checking{}$, their inverse,
and {$O(\log(1/\delta)/\sqrt{\eps})$} other basic operations (independent of $M$) on the internal register.

More precisely, an observation of the internal register
when the algorithm declares that $M\neq\emptyset$,
 outputs a random $x\in M$ with probability $|\alpha_x|^2/P_M$. 
\end{theorem}

\subsection{Distributed quantum optimization}\label{prelim:ampl}  
Let us first explain how to derive a generic quantum procedure for optimization problems in the centralized model.
Indeed, a well known application of amplitude amplification is optimization, such as minimum finding~\cite{Durr+SICOMP06}.
We are now given a procedure $\eval$ instead of $\checking$, which evaluates the (unknown) function we want to maximize.
\begin{description}
\item[$\eval{}$:] Performs the transformation 
$$\ket{x,0}_I\ket{\text{data}(x)}\ket{\text{init}} \mapsto 
\ket{x,f(x)}_I\ket{\text{data}(x)}\ket{\text{init}},$$
where $f\colon X\to\mathbb{Z}$ is the function we want to maximize.
\end{description}
Define the probability $P_{\opt}$ to observe an element where~$f$ takes a maximum value when measuring Register $I$ after one application of $\setup{}$ on the initial state as $$P_{\opt}=\sum_{x : \text{$f(x)$ is maximal}}|\alpha_x|^2.$$ 
A standard application of amplitude amplification shows how to maximize $f$ using $\Theta(1/\sqrt{P_{\opt}})$ iterations of $\setup{}$ and $\eval{}$. We state this result as follows.
\begin{corollary}[Quantum optimization]\label{cor:qo}
Let $\eps>0$ be such that $P_{\opt}\geq \eps$. 
Then, for any $\delta>0$, there is a quantum algorithm that can find,with probability at least $1-\delta$, some element $x$ such that $f(x)$ is maximum,
using one unitary operator $\init$,
{$O(\log(1/\delta)/\sqrt{\eps})$} unitary operators $\setup{}$ and $\eval{}$, and their inverses,
and {$O(\log(1/\delta)/\sqrt{\eps})$} other basic operations (independent of~$f$) on the internal register. 
\end{corollary}
\begin{proof}
The proof uses standard applications of Theorem~\ref{thm:aa} such as in~\cite{Durr+SICOMP06}. We give the main lines of the reduction:
\begin{quote}
\begin{enumerate}
\item 
Start with some fixed $a\in X$;
\item
\label{lqo} Use amplitude amplification of Theorem~\ref{thm:aa}
with $\eps'=1/2$ and $\delta'=\delta$ to find $b>a$;
\item 
If an element is found, then set $a=b$ and go 
to (\ref{lqo});
\item 
Else if $\eps'>\eps$, then set $\eps'=\eps'/2$ and go 
to (\ref{lqo});
\item 
Output $a$ and stop.
\end{enumerate}
\end{quote}
The expected number of iterations, that is of jumps to Step~(\ref{lqo}), can be shown to be $O(\log(1/\eps))$,
with an expected number of applications of unitary maps and of other basic operations in {$O(\log(1/\delta)/\sqrt{\eps})$}
using similar arguments to~\cite{Durr+SICOMP06} for minimum finding. 
In order to get the claimed worst case bounds, one just need to abort
the computation when too much resources have been used and to output the current value of $a$.
\end{proof}

We now explain how to implement quantum optimization in the distributed model. In the distributed setting, the internal register $I$ and therefore the control of the algorithm itself are simply centralized by a leader node. The two other registers $\ket{\text{data}(x)}$ and $\ket{\text{init}}$, however, can be distributed among all the nodes of the network, and implementing Procedures $\setup$ and $\eval$ thus generally requires communication through the network. As stated in Corollary \ref{cor:qo}, besides calls to $\setup$ and $\eval$ and their inverses, all the other operations performed in quantum optimization are done exclusively on the internal register $I$, and then can be implemented locally (i.e., without communication) by the chosen leader node. The round complexity of distributed quantum optimization thus depends on the round complexity of Procedures $\setup$ and $\eval$. The precise statement follows.

\begin{theorem}[Distributed quantum optimization]\label{thm:dqo}
Let $G=(V,E)$ be a distributed network with a predefined node $\leader\in V$.
Assume that $\init$ can be implemented within $T_0$ rounds and using  $\mem$ memory per node
in the quantum CONGEST model, and that
unitary operators $\setup{}$ and $\eval{}$ and their inverses
can be implemented within $T$ rounds and using $\mem$ memory per node. Assume that $s=\Omega(\log |X|)$.

Let $\eps>0$ be such that $P_{\opt}\geq \eps$. 
Then, for any $\delta>0$, the node $\leader$ can find, with probability at least $1-\delta$, some element $x$ such that $f(x)$ is maximum,
in $T_0+{O(\log(1/\delta)/\sqrt{\eps})} \times T$ rounds.
This quantum algorithm uses $O(\mem \times \log(1/\eps))$ memory per node. 
\end{theorem}
\begin{proof}
The round complexity follows immediately from Corollary \ref{cor:qo}. Let us now analyze the amount of memory needed for the node $\leader$. Since measurements in the procedure described in Corollary \ref{cor:qo} are only performed at the end of the computation, all outcomes of amplitude amplification in Line~(\ref{lqo}) are recorded in the procedure, leading to an additional term $O(\log|X|\times \log(1/\eps))$ in the memory size. The total amount of memory needed for the implementation is thus $O(s+\log|X|\times \log(1/\eps))=O(\mem \times \log(1/\eps))$ qubits for $\leader$ and $O(s)$ qubits for the other nodes.
\end{proof}

\section{First Algorithm : Exact Computation of the Diameter}\label{sec:alg1}
Let $G=(V,E)$ denote the network we consider. We write $n=|V|$ and use $D$ to denote its diameter. We assume that the network~$G$ has elected a node $\leader\in V$, and computed its eccentricity $\ecc(\leader)$. This can be done using standard methods in $O(D)$ classical rounds and $O(\log n)$ memory space per node.
In addition, our algorithm will use a Breadth First Search tree rooted at $\leader$ (denoted $\bfs(\leader)$), which can be computed efficiently as described in the following proposition.
\begin{proposition}\label{prop:init}
There is a classical procedure that allocates to each node $v\in V$
its parent on $\bfs(\leader)$ and its distance to $\leader$
in $O(D)$ rounds and $O(\log n)$ memory space per node.
\end{proposition}
\begin{proof}
The procedure is the classical procedure described in Figure \ref{fig:init} that performs a Breadth First Search tree from $\leader$. 
This procedure has complexity $O(D)$ and can be implemented with $O(\log n)$ bits of memory per node.
%
\end{proof}

\begin{figure}[t!]
\begin{center}
\fbox{
\begin{minipage}{14cm} 
\textbf{In $\boldsymbol{\ecc(\leader)=O(D)}$ rounds:}
\begin{enumerate}
\item First, only $\leader$ is activated, it declares itself as its own parent, and
sends a message through all its edges in order to activate its neighbors for the next round.
\item When a message reaches  node $v$ from node $u$:
\begin{itemize}
\item If $v$ was already activated, then $v$ ignores the message. 
\item Else $v$ becomes activated, 
it sets its parent to $u$ and
sends, through all its edges, 
its distance to $\leader$ in order to activate its neighbors for the next round.
\end{itemize}
\end{enumerate}
\end{minipage}
}
\end{center}\vspace{-4mm}
\caption{Construction of $\bfs(\leader)$ in Proposition \ref{prop:init}.}\label{fig:init}
\end{figure}

Our quantum algorithm will compute the diameter, i.e., the maximum eccentricity among all nodes, using quantum optimization. In Section \ref{sub:alg-simple}, we first give a simpler quantum algorithm that computes the diameter in $O(\sqrt{n}D)$ rounds.  In Section \ref{sub:alg-complicated} we then give our $O(\sqrt{nD})$-round quantum algorithm.
\vspace{2mm}

\subsection{A simpler quantum algorithm}\label{sub:alg-simple}
We apply the framework of Section \ref{prelim:ampl} with $X=V$ and the function $f\colon V\to \mathbb{Z}$ defined as 
\begin{equation}\label{eq:f1}
f(u)=\ecc(u)
\end{equation}
for all $u\in V$.
Obviously, maximizing $f$ gives the diameter of the network. 

Procedure $\init$ is the classical procedure of Proposition~\ref{prop:init}. This means that $\ket{\text{init}}$ represents the classical information computed by this classical procedure: this is a quantum register shared by all the nodes of the network, in which the part owned by each node contains its parent on $\bfs(\leader)$ and its distance to $\leader$. 

For any $u_0\in V$, the quantum state $\ket{\text{data}(u_0)}$ used in the quantum algorithm is defined as
$
\ket{\text{data}(u_0)}=\bigotimes_{v\in V}\ket{u_0}_{v}.
$
We set $\alpha_{u_0}=1/\sqrt{n}$ for all $u_0\in V$, and thus have $P_\opt\ge 1/n$. 
We first observe that Procedure $\setup$ can be implemented efficiently.
\begin{proposition}\label{prop:setup}
The procedure $\setup{}$ can be implemented in $O(D)$ rounds and using $O(\log n)$ memory space per node.
\end{proposition}
\begin{proof}
The node $\leader$ first prepares the quantum state
\[
\frac{1}{\sqrt{n}}\sum_{u_0\in V}\ket{u_0}_{\leader}.
\]
Then this state is simply broadcast using CNOT copies
to all nodes of the network along $\bfs(\leader)$, which takes
$d$ rounds. 
The resulting state is
\[
\frac{1}{\sqrt{n}}\sum_{u_0\in V}\ket{u_0}_{\leader}\bigotimes_{v\in V}\ket{u_0}_{v}=\frac{1}{\sqrt{n}}\sum_{u_0\in V}\ket{u_0}_{\leader}\ket{\text{data}(u_0)},
\]
as required.
\end{proof}

We now observe that Procedure $\eval$ can also be implemented efficiently.

\begin{proposition}\label{prop:eval1}
The procedure $\eval{}$ for the function $f$ defined in $(\ref{eq:f1})$ can be implemented in $O(D)$ rounds and using $O(\log n)$ memory space. 
\end{proposition}
\begin{proof}
Consider the following problem: all the nodes of the network receive as input $u_0$, for some fixed $u_0\in V$, and the goal is for node $\leader$ to output $\ecc(u_0)$.
There is a straightforward $O(D)$-round classical algorithm solving this problem using $O(\log n)$ bits of memory per node: node $u_0$ starts building a Breadth First Search tree rooted at $u_0$, the nodes use this tree to compute their distance to $u_0$, and finally $\ecc(u_0)$ is transmitted to the node $\leader$. Note that this algorithm is precisely a procedure to compute $f(u_0)$ when the nodes initially share the classical information contained in $\ket{\text{data}(u_0)}$ (i.e., the nodes all know $u_0$). As mentioned in Section \ref{app:dqo}, this procedure can then been quantized using standard techniques in order to obtain a quantum algorithm that performs the transformation
\[
\ket{u_0,0}_I\ket{\text{data}(u_0)}\ket{\text{init}} \mapsto \ket{u_0,f(u_0)}_I\ket{\text{data}(u_0)}\ket{\text{init}},
\]
with the same round and space complexities.
\end{proof}

We are now ready to analyze the overall complexity of our algorithm.
We apply Theorem~\ref{thm:dqo} to maximize the function $f$ defined in~$(\ref{eq:f1})$. Propositions~\ref{prop:init},~\ref{prop:setup} and~\ref{prop:eval1}, along with the lower bound $P_\opt\ge 1/n$, show that this quantum algorithm uses $O(\sqrt{n}D)$ rounds and $O((\log n)^2)$ memory per node for any constant error probability~$\delta$.

\subsection{The final quantum algorithm}\label{sub:alg-complicated}
We now show how to reduce the complexity to $O(\sqrt{nD})$. The key idea is to modify the function to be maximized, in order to increase $P_\opt$ when $d$ is large. Before defining the new function, let us introduce some definitions.

%
We take some integer~$d$ such that $d\leq D\leq 2d$. Since $\ecc(v)\le D\le 2\times \ecc(v)$ for any node $v\in V$, this will concretely be done by choosing $d=\ecc(\leader)$. 
Similarly to \cite{Peleg+ICALP12}, we introduce the following numbering on the vertices of $G$ via a Depth-First-Search traversal of $\bfs(\leader)$.
\begin{definition}[$\dfs$-numbering]
The \emph{$\dfs(\leader)$-number} of $v\in V$, denoted by $\tau(v)$,
is the length of the path to reach $v$ from $\leader$ on a 
$\dfs$-traversal of $\bfs(\leader)$.
In particular 
$\tau(\leader)=0$.
\end{definition}

We now use the $\dfs$-numbering to define sets of vertices. 

\begin{definition}\label{def:sets}
For any $u\in V$,
define $S(u)=\{v\in V : \tau(v)\in [\tau(u),\tau(u)+2d \mod{2n}]\}$.
Let $\Ss$ be the collection of subsets $S(u)$ for all $u\in V$.
\end{definition}

We following lemma, which will be crucial for our analysis, shows that $\Ss$ covers well the set $V$.
\begin{lemma}\label{lem:success}
Fix any  $v\in V$, and take some $u_0\in V$ uniformly at random.
Then $\Pr[v\in S(u_0)]\geq \frac{d}{2n}$.
\end{lemma}
\begin{proof}
Consider a sequential execution $\pi$ of $\dfs(\leader)$ on $\bfs(\leader)$.
We now consider $\pi$ as a circle by attaching its extremities.
Observe that each top-down move to some node $v$ at step $t$ causes the numbering $\tau(v)=t$. 
In addition, since $\bfs(\leader)$ has depth $d$, any segment of $\pi$ with $m_d$ 
top-down moves and $m_u$ bottom-up ones must satisfy $|m_d-m_u|\leq d$.

Thus for any vertex $v\in V$, consider the part of $\pi$ of length $2d$ ending in a
top-down move to $v$. It must contain at least $\ceil{d/2}$ top-down moves to some
vertices $u_1,\ldots, u_{\ceil{d/2}}$. Therefore $v\in S(u_i)$, for $i=1,\ldots,\ceil{d/2}$,
and we conclude that $\Pr[v\in S(u_0)]\geq \frac{d}{2n}$.
\end{proof}

The function we will now optimize is the function $f\colon V\to \mathbb{Z}$ defined as 
\begin{equation}\label{eq:f2}
f(u)=\max_{v\in S(u)} (\ecc(v))
\end{equation}
for any $u\in V$. Obviously, maximizing $f$ gives the diameter of the network. We optimize this function via the distributed quantum optimization procedure of Section~\ref{prelim:ampl}.  We set $X$, $\alpha_x$, $\ket{\text{init}}$ and $\ket{\text{data}(x)}$ exactly as in Section \ref{sub:alg-simple}. Lemma \ref{lem:success} immediately implies that $P_\opt\ge d/2n$.

The procedures $\init$ and $\setup$ are the same as in Section \ref{sub:alg-simple}. The procedure $\eval{}$, however, is different since the function to evaluate is not the same. We describe the new procedure, which is based on a refinement of the deterministic distributed algorithm for diameter in~\cite{Peleg+ICALP12}, in the next proposition. The proof of this proposition is postponed to Section \ref{sub:alg-eval}.

\begin{proposition}\label{prop:eval2}
The procedure $\eval{}$ for the function $f$ defined in $(\ref{eq:f2})$ can be implemented in $O(D)$ rounds and using $O(\log n)$ memory space. 
\end{proposition}

We are now ready to analyze the overall complexity of our algorithm and prove Theorem \ref{th:UB1}.
\begin{proof}[Proof of Theorem \ref{th:UB1}]
We apply Theorem~\ref{thm:dqo} to maximize the function $f$ defined in~$(\ref{eq:f2})$. Propositions~\ref{prop:init},~\ref{prop:setup} and~\ref{prop:eval2}, along with the lower bound $P_\opt\ge d/2n$, show that this quantum algorithm uses $O(\sqrt{nD})$ rounds and $O((\log n)^2)$ memory per node for any constant error probability $\delta$.
\end{proof}

\subsection{The evaluation procedure}\label{sub:alg-eval}
We prove Proposition \ref{prop:eval2} by describing the Procedure $\eval$ for the function $f$ defined in~$(\ref{eq:f2})$. Remember that this procedure should, for any $u_0\in V$ given as input to all the nodes of the network, enable the node $\leader$ to output $f(u_0)=\max_{v\in S(u_0)} (\ecc(v))$. Remember also that we are assuming the each node $v\in V$ knows its parent on $\bfs(\leader)$ and its distance to $\leader$ (this information was computed in Procedure $\init$). 

The procedure $\eval$ is described in Figure \ref{fig:alg-eval}. (The procedure is described as a classical procedure. As mentioned in Section~\ref{app:dqo} and in the proof of Proposition \ref{prop:eval1}, it can then been quantized using standard techniques.)

The main idea is simple: each node $v\in S(u_0)$ starts a process
in order to globally compute its eccentricity.
This will be done in parallel, and the main difficulty is to avoid congestions.
We achieve this goal by refining techniques used in the deterministic distributed algorithm from~\cite{Peleg+ICALP12}. As in \cite{Peleg+ICALP12}, we coordinate the starts of all the processes so that messages initiated by nodes $v\in S(u_0)$ arrive in the same order as their $\dfs(\leader)$-numbers. Therefore, there is no congestion, and messages to disregard are easily identified, even with a local memory of logarithmic size for each node. 

A difference with \cite{Peleg+ICALP12} is that we do not use the $\dfs(\leader)$-number $\tau(v)$, which is too costly to compute, but instead the quantity $\tau'(v)=\tau(v)-\tau(u_0)$, which is much easier to compute, to coordinate the starts of the process. More precisely, node $v\in S(u_0)$ starts its process at round $2\tau'(v)$, see Step 2(2) of the algorithm. 

Another difference is that each node $v$ in Step 2 of the algorithm only keeps in memory the maximum of their previously computed distances. Nodes receive at each round messages of the form $(\tau',\delta)$ that they either disregard or broadcast after incrementing $\delta$. This message has been initiated at round $\tau'$ by the node $u\in S$ such that $\tau'=\tau'(u)$  in Step 2(2).
The number $\delta$ is the number of rounds this message took to reach $v$
for the first time, i.e., $\delta=d(u,v)$.
The value~$t_v$ is used to decide which messages can be disregarded. The value~$d_v$ represents the maximum distance $d(u,v)$ over all the nodes $u$ processed so far.

\begin{figure}[t!]
\begin{center}
\fbox{
\begin{minipage}{14cm} 
\begin{description}
\item[Step 1.] The network performs sequentially $2d$ steps of a Depth First Search traversal on $\bfs(\leader)$ starting at $u_0$ (if it reaches the end of the DFS, it starts again from $\leader$). Let us write $S$ the nodes visited by the process, and for each $v\in S$, let $\tau'(v)$ to be the index of the round they are reached (for the first time) by the traversal.
\item[Step 2.] Each node $v\in V$ implements the following process during $6d$ rounds:
\begin{enumerate}
\item Set  $t_v=-1$ and $d_v=0$.
\item If $v\in S$ and only at round $2\tau'(v)$:\\
Broadcast message $(\tau'(v),0$) along all edges $(v,w)$. 
\item At each round: 
\begin{enumerate}
\item Disregards all messages of type $(\tau',\delta)$ with $\tau'\leq t_v$\label{clear}.
\item Keep at most one of the remaining messages of type $(\tau',\delta)$ with $\tau' > t_v$  \label{no-congestion}.\\
\emph{(In fact they are all equal according to Lemma~\ref{lem:no-congestion}.)}\\
If one message has been kept then:\\
$\phantom{set}$Set $t_v=\tau'$ and $d_v=\max(d_v,\delta)$.\\
$\phantom{set}$Broadcast message $(\tau',\delta+1)$ along all edges $(v,w)$ of $G$.
\end{enumerate}
\end{enumerate}
\item[Step~3.] Each node $v\in V$ sends its current value $d_v$ to $\leader$. 
In order to avoid congestions, the transmission is done bottom up on $\bfs(\leader)$, and at each node only the maximum of received values is transmitted.
\item[Step~4.] Node $\leader$ computes the maximum of all received values. 
\item[Step~5.] Revert steps 3 to 1, in order to clean all registers.
\end{description}
\end{minipage}
}
\end{center}\vspace{-4mm}
\caption{The procedure $\eval$.}\label{fig:alg-eval}
\end{figure}

The procedure $\eval$ can be implemented in $O(D)$ rounds and using $O(\log n)$ space per node since at any time each node only needs to keep one message $(\tau',\delta)$, in addition to $t_v$ and $d_v$, in memory. The correctness follows from the following three lemmas, which together show that the procedure computes the maximum of $\ecc(v)$ over all $v\in S(u_0)$.
\begin{lemma}\label{lem:su0}
At Step~1, the computed set $S$ satisfies $S=S(u_0)$.
Furthermore,
for every nodes $v,w\in S$, if $\tau'(v)<\tau'(w)$ then $d(v,w)\leq \tau'(w)-\tau'(v)$. 
\end{lemma}
\begin{proof}
First observe that at Step~1, $\tau'(v)=\tau(v)-\tau(u_0)$, for all $v\in S(u_0)$.
Therefore, the computed set $S$ satisfies $S=S(u_0)$.
Then the rest of the statement is a direct consequence of~\cite[Property~1]{Peleg+ICALP12}
which was stated for the $\dfs(\leader)$-numbering $\tau$, and therefore remains valid for $\tau$ since $\tau'(v)=\tau(v)-\tau(u_0)$, for all $v\in S(u_0)$.
%
\end{proof}

\begin{lemma}\label{lem:increasing}
Let $r_1$ be the first round that a message of
type $(\tau'_1,*)$ reaches a given node $v$,
and similarly for $r_2$ and $\tau'_2$.
Then $r_2-r_1\geq \tau_2'-\tau_1'$.
In particular, if $\tau_1'<\tau_2'$ then $r_1<r_2$.
\end{lemma}
\begin{proof}
Let $z_1,z_2$ be such that $\tau_i'=\tau'(z_i)$, for $i=1,2$.
This is the first time that a message from $z_i$ reaches $v$, 
therefore $r_i=2\tau'(z_i)+d(z_i,v)$.
In addition, by the triangle inequality, we also have that
$d(z_1,v)-d(z_2,v)\leq d(z_1,z_2)$.
Thus
$$d(z_1,z_2)\geq d(z_1,v)-d(z_2,v)=r_1-r_2+2(\tau'(z_2)-\tau'(z_1)).$$

Since $\tau'(z_1)<\tau'(z_2)$, then 
by  Lemma~\ref{lem:su0} we also have that
$$d(z_1,z_2)\leq 
\tau'(z_2)-\tau'(z_1).$$

Therefore we get
\[
r_2-r_1\geq \tau'(z_2)-\tau'(z_1)>0,
\]
which proves the lemma.
\end{proof}

\begin{lemma}\label{lem:no-congestion}
At Step~\ref{no-congestion}, all remaining messages
are identical.
\end{lemma}
\begin{proof}
Assume by contradiction that at least two messages $(\tau_1',\delta_1)$ and $(\tau_2',\delta_2)$ remain
such that either $\tau_1'\neq \tau_2'$ or $\delta_1\neq \delta_2$. 
By construction of our algorithm, they are the first messages of types $(\tau_1',*)$ and $(\tau'_2,*)$ reaching $v$.

If $\tau'_1\neq \tau'_2$,
by Lemma~\ref{lem:increasing}, those messages cannot arrive at the same round, which contradicts the fact they both arrive at the current round.

If $\tau'_1=\tau'_2$, but $\delta_1\neq \delta_2$, we get another contradiction. Indeed, those messages come from the same vertex $u$ such that $\tau'(u)=\tau'_1$, and reach $v$ at the same round for the first time. Thus they must satisfy $\delta_1=\delta_2$.
\end{proof}

\section{Second Algorithm: 3/2-Approximation of the Diameter}
The framework for distributed quantum optimization we developed in this paper, and in particular its application to find vertices of maximum eccentricities presented in Section~\ref{sec:alg1}, can be used as well to design fast quantum approximation algorithms for the diameter. In this section we show how to apply our framework to speed up the $3/2$ classical approximation algorithm designed in~\cite{Holzer+DISC14}, which computes a value~$\bar{D}$ such that $\bar{D}\leq D\leq 3\bar{D}/2$, and prove Theorem~\ref{th:UB2}.

Our quantum algorithm for $3/2$-approximation is described in Figure \ref{fig:alg-approx}. In the description of the algorithm, $s\in[n]$ is an integer parameter that will be set later.

\begin{figure}[t!]
\begin{center}
\fbox{
\begin{minipage}{14cm} 
\hspace{-10mm}
\begin{description}
\item[Preparation in $\boldsymbol{\tilde{O}(n/s + D)}$ rounds.] 
This part is the same as Steps~1 to~5 of Algorithm~1 in \cite{Holzer+DISC14}:
\begin{enumerate}
\item Each vertex joins a set $S$ with probability $(\log n)/s$.\\
If more than $(n(\log n)^2/s)$ vertices join $S$, then abort.
\item Each vertex $v\in V$ computes the closest node in $S$ to $v$, which is denoted $p(v)$.\\
Compute a node $w\in V$ maximizing the distance $d(w,p(w))$.
\item Compute a $\bfs$ tree from $w$.\\
The $s$ closest nodes to $w$ in the tree join a set $R$. 
\end{enumerate}
\item[Quantum optimization $\boldsymbol{\tilde O(\sqrt{sD}+D)}$ rounds.] Compute the maximum eccentricity of the vertices in $R$. 
\item[Output.] Return the maximum computed eccentricity.
\end{description}
\end{minipage}
}
\end{center}\vspace{-4mm}
\caption{Our quantum algorithm computing a $3/2$-approximation of the diameter.}\label{fig:alg-approx}
\end{figure}

This quantum algorithm is exactly the same as Algorithm 1 in~\cite{Holzer+DISC14}, except for the second part: in \cite{Holzer+DISC14} the maximum eccentricity was computed classically in $O(s+D)$ rounds by computing a $\bfs$ tree for each vertex in $R$. We thus omit the details about the implementation of the first part and the proof of correctness of the whole algorithm (which can be found in \cite{Holzer+DISC14}). Let us nevertheless point out that this first part requires a polynomial amount of (classical) memory. The quantum procedure for the second part we describe below, however, will only require a polylogarithmic amount of memory.

The quantum procedure for the second part is essentially the same as the quantum algorithm presented in Section \ref{sec:alg1}, which computed the maximum eccentricity of all the nodes in $V$ (i.e., the diameter of the network). The difference is that this time we do the optimization only on the vertices in $R$: the function to optimize the function of Equation $(\ref{eq:f2})$ restricted to the domain $R\subseteq V$. By replacing node $\leader$ by node $w$, and replacing ``\textrm{mod} $2n$'' by ``\textrm{mod}~$2s$'' in Definition \ref{def:sets}, we obtain $P_\opt\ge d/2s$ instead $P_\opt\ge d/2n$ since~$R$ is defined as the set of the $s$ closest nodes to $w$. By additionally modifying the procedure $\setup$ so that it distributes the state $\frac{1}{\sqrt{s}}\sum_{u_0\in R} \bigotimes_{v\in V}\ket{u_0}_{v}$ instead of $\frac{1}{\sqrt{n}}\sum_{u_0\in V} \bigotimes_{v\in V}\ket{u_0}_{v}$, Theorem~\ref{thm:dqo} implies that the maximum eccentricity of the vertices in $R$ can be computed with high probability in $\tilde O(\sqrt{sD}+D)$ rounds and using $O((\log n)^2)$ qubits of memory per node.

The overall round complexity of our algorithm is $\tilde O(n/s + \sqrt{ns} +D)$.
Choosing $s=\Theta({n}^{2/3}D^{-1/3})$ gives complexity $\tilde O(\sqrt[3]{nD}+D)$, as claimed in Theorem~\ref{th:UB2}.


\section{Lower Bound: Tight Bound for Small Diameter}\label{sec:LB1}
In this section we show that the connection established in prior works (in particular in \cite{Frischknecht+SODA12,Holzer+PODC12,Abboud+DISC16,Bringmann+DISC17}) between the round complexity of computing the diameter and the two-party communication complexity of the disjointness function still holds in the quantum setting. We then use this connection to prove Theorem \ref{th:LB1}. 

In the following we will use the notation $G=(U,V,E)$ to denote a bipartite graph whose partition has parts $U$ and $V$: it will always be implicitly assumed that the sets $U$ and $V$ are disjoint, and that each edge in $E$ has an endpoint in $U$ and its other endpoint in $V$. We will use $\Delta(G)$ to denote the largest distance between one vertex in $U$ and one vertex of $V$. We will use $\Ee(U)$ to denote the set of all pairs of elements in $U$ and $\Ee(V)$ to denote the set of all pairs of elements in $V$. 

All the recent lower bounds on the classical round complexity of computing or approximating the diameter rely on the notion of reduction introduced in the next definition. 

\begin{definition}\label{def:sim}
Let $b,k,d_1$ and $d_2$ be four functions from $\nat$ to $\nat$. A $(b,k,d_1,d_2)$-reduction from disjointness to diameter computation is a family of triples $\{(G_n,g_n,h_n)\}_{n\ge 1}$ such that, for each $n\ge 1$,  $G_n=(U_n,V_n,E_n)$ is a bipartite graph with $|U_n\cup V_n|=n$ and $|E_n|=b$, $g_n$ is a function from $\{0,1\}^k$ to $\Ee(U_n)$ and $h_n$ is a function from $\{0,1\}^k$ to $\Ee(V_n)$ satisfying the following two conditions for any $x,y\in\{0,1\}^k$:
\begin{itemize}
\item[(i)]
if $\DISJ{k}(x,y)=1$ then  $\Delta(G_n(x,y))\le d_1$,
\item[(ii)]
if $\DISJ{k}(x,y)=0$ then $\Delta(G_n(x,y))\ge d_2$,
\end{itemize}
where $G_n(x,y)$ denotes the bipartite graph obtained from $G_n$ by adding an edge $\{w,w'\}$ for each $\{w,w'\}\in g_n(x)\cup h_n(y)$.
\end{definition}

In particular we will use in this paper the following two constructions from prior works.
\begin{theorem}[\cite{Holzer+PODC12}]\label{th:construction1}
There exists a $(\Theta(n),\Theta(n^2),2,3)$-reduction from disjointness to diameter computation.
\end{theorem}
\begin{theorem}[\cite{Abboud+DISC16}]\label{th:construction2}
There exists a $(\Theta(\log n),\Theta(n),4,5)$-reduction from disjointness to diameter computation.
\end{theorem}
We include a proof of Theorem \ref{th:construction1} as an example to illustrate the notations of Definition~\ref{def:sim}.
\begin{figure}
\centering
\begin{tikzpicture}[scale=0.3,rectnode/.style={shape=circle,draw=black,minimum height=10mm},roundnode/.style={circle, draw=black!60, fill=blue!5, very thick, minimum size=5mm}]
    \newcommand\XA{1.5}
    \newcommand\YA{2}
    
    \node[roundnode,draw=black] (ll2) at (-3*\XA,1*\YA) {\small $\ell'_1$};    
    \node[roundnode,draw=black] (ll1) at (-3*\XA,3*\YA) {\small $\ell'_0$}; 
    \node[roundnode,draw=black] (l2) at (-3*\XA,6*\YA) {$\ell_1$}; 
    \node[roundnode,draw=black] (l1) at (-3*\XA,8*\YA) {$\ell_0$}; 
    \node[roundnode,draw=black] (a) at (-9*\XA,11*\YA) {$a$}; 
    \node[roundnode,draw=black] (rr2) at (3*\XA,1*\YA) {\small $r'_1$};    
    \node[roundnode,draw=black] (rr1) at (3*\XA,3*\YA) {\small $r'_0$}; 
    \node[roundnode,draw=black] (r2) at (3*\XA,6*\YA) {$r_1$}; 
    \node[roundnode,draw=black] (r1) at (3*\XA,8*\YA) {$r_0$}; 
    \node[roundnode,draw=black] (b) at (9*\XA,11*\YA) {$b$}; 
    \path[] 
    (a) edge (b)
     (a) edge (l1)
     (a) edge (l2)
     (a) edge (ll1)
     (a) edge (ll2)
     (ll1) edge (ll2)
     (l1) edge (l2)
	 (b) edge (r1)
     (b) edge (r2)
     (b) edge (rr1)
     (b) edge (rr2)
     (rr1) edge (rr2)
     (r1) edge (r2)
     (l1) edge (r1)
     (l2) edge (r2)
     (ll1) edge (rr1)
     (ll2) edge (rr2);
    \draw[dashed] (-6.5,0) -- (-2.5,0) -- (-2.5,8) -- (-6.5,8) -- (-6.5,0);
    \draw[dashed] (-6.5,10) -- (-2.5,10) -- (-2.5,18) -- (-6.5,18) -- (-6.5,10);
    \draw[dashed] (6.5,0) -- (2.5,0) -- (2.5,8) -- (6.5,8) -- (6.5,0);
    \draw[dashed] (6.5,10) -- (2.5,10) -- (2.5,18) -- (6.5,18) -- (6.5,10);
    \node[draw=none,fill=none] at (-4.5,-1) {$L'$};
    \node[draw=none,fill=none] at (4.5,-1) {$R'$};
    \node[draw=none,fill=none] at (-4.5,19) {$L$};
    \node[draw=none,fill=none] at (4.5,19) {$R$};
   
\end{tikzpicture}
\caption{The graph $G_n$ used in the proof of Theorem \ref{th:construction1}, for $n=2$.}\label{fig:graph1}
\end{figure}
\begin{proof}[Proof of Theorem \ref{th:construction1}]
We define the graph $G_n=(U_n,V_n,E_n)$ as follows. Here we assume, for simplicity, that $n-2$ is a multiple of 4 (the construction can be easy modified, by adding dummy nodes, to deal with the case when $n$ is not of this form). We define $s=(n-2)/4$.  The construction is illustrated for $n=10$ in Figure \ref{fig:graph1}.

We have $U_n=L\cup L'\cup \{a\}$ and $V_n=R\cup R'\cup \{b\}$, where $L=\{\ell_0,\ldots,\ell_{s-1}\}$, $L'=\{\ell'_0,\ldots,\ell'_{s-1}\}$, $R=\{r_0,\ldots,r_{s-1}\}$ and $R'=\{r'_0,\ldots,r'_{s-1}\}$ are four sets each containing $s$ nodes, and $a,b$ are two additional nodes. The set $E_n$ is defined as follows: we put edges to make each of $L$, $L'$, $R$ and $R'$ an $s$-clique. Moreover, for each $i\in\{0,\ldots,s-1\}$, we add an edge from $\ell_i$ to $r_i$ and an edge from $\ell'_i$ to $r'_i$. Finally, node $a$ is connected by an edge to each node in $L\cup L'$, while node $b$ is connected by an edge to each node in $R\cup R'$. Nodes $a$ and $b$ are also connected by an edge. Note that $|E_n|=2s+1=\Theta(n)$.

Now we define the functions $g_n$ and $h_n$. Given any binary string $x\in \{0,1\}^{s\times s}$, we define $g_n(x)$ as the set of all pairs $\{\ell_i,\ell'_j\}$ with $i,j\in\{0,\ldots,s-1\}$ such that $x_{i,j}=0$. Similarly, given any binary string $y\in \{0,1\}^{s\times s}$, we define $h_n(y)$ as the set of all pairs $\{r_i,r'_j\}$ with $i,j\in\{0,\ldots,s-1\}$ such that $y_{i,j}=0$. This means that the graph $G_n(x,y)$ is obtained from $G_n$ by putting an additional edge between node $\ell_i$ and node $\ell'_{j}$ if and only if $x_{i,j}=0$, and similarly by putting an additional edge between node $r_i$ and node $r'_{j}$ if and only if $y_{i,j}=0$. 

It is easy to check that conditions (i) and (ii) of Definition \ref{def:sim} hold with $k=s^2=\Theta(n^2)$, $d_1=2$ and $d_2=3$, by observing that the distance in $G_n(x,y)$ between $\ell_i$ and $r'_j$ (and, symmetrically, the distance between $\ell'_j$ and $r_i$ as well) is~3 if $x_{i,j}=y_{i,j}=1$, and $2$ otherwise.
\end{proof}
An easy argument, used in all these prior works mentioned above, shows that the existence of a $(b,k,d_1,d_2)$-reduction from disjointness to diameter computation implies a $\tilde \Omega(k/b)$ lower bound on the round complexity of any classical distributed algorithm that decides whether the diameter of the network is at most~$d_1$ or at least~$d_2$. The main technical contribution of this section is the following quantum version of this argument, which is proved using the recent result from \cite{BGK+15} on the message-bounded quantum communication complexity of disjointness presented in Section \ref{sec:prelim-qcc}.
\begin{theorem}\label{th:red1}
Assume that there exists a $(b,k,d_1,d_2)$-reduction from disjointness to diameter computation. Then any quantum distributed algorithm that decides, with high probability, whether the diameter of the network is at most~$d_1$ or at least $d_2$ requires $\tilde \Omega(\sqrt{k/b})$ rounds.
\end{theorem}
\begin{proof}
Assume that there exists an $r$-round quantum algorithm $\Aa$ that decides if the diameter of a network is at most~$d_1$ or at least $d_2$. We will use this algorithm to construct a protocol that computes the disjointness function $\DISJ{k}$ in the two-party quantum communication complexity setting. Let $x\in \{0,1\}^{k}$ and $y\in \{0,1\}^{k}$ be Alice's input and Bob's input, respectively. 

The protocol works as follows. First note that Alice and Bob can jointly construct the graph $G_n(x,y)$: Alice can construct the subgraph induced by $U_n$ (which depends on $x$), while Bob can construct the subgraph induced by $V_n$ (which depends on $y$). Alice and Bob will thus simply simulate the computation of $\Aa$ on $G_n(x,y)$. To do this, they only need to exchange messages corresponding to the communication occurring along the $b$ edges of $G_n(x,y)$ at each round of Algorithm $\Aa$. This can be done by replacing one round of communication in $\Aa$ by two messages of $O(b\log n)$ qubits (one from Alice to Bob, and the other from Bob to Alice). At the end of simulation Alice and Bob obtain the result of the computation by $\Aa$ on $G_n(x,y)$. From our assumption on $\Aa$, this enables them to decide if the diameter of $G_n(x,y)$ is at most~$d_1$ or at least $d_2$, and thus to compute $\DISJ{k}(x,y)$. The whole simulation uses $2r$ messages and a total number of $O(rb\log n)$ qubits of communication. From Theorem~\ref{th:BGK+15}, we conclude that
$
rb=\tilde \Omega\left(\frac{k}{r}+r\right),
$
which implies $r=\tilde \Omega(\sqrt{k/b})$. 
\end{proof}
Theorem \ref{th:LB1} directly follows by combining Theorem \ref{th:construction1} and Theorem \ref{th:red1}.

\section{Lower Bound: Case of Large Diameter}\label{sec:LB2}
In this section we  prove Theorem \ref{th:LB2}. We first prove in Section \ref{sec:sim} a general simulation result that will be crucial for our analysis, and then give the proof of Theorem~\ref{th:LB2} in Section \ref{sec:proofLB2}.

\subsection{Simulation argument}\label{sec:sim}
For any $d\ge 1$, let us introduce the following network $\Gd$. The network consists of two nodes $A$ and~$B$ connected by a path of length $d+1$, corresponding to $d$ intermediate nodes $P_1,\ldots,P_{d}$. The total number of nodes is thus $d+2$, and the number of edges is $d+1$. Each edge of $\Gd$ is a quantum channel of bandwidth $\band$ qubits. The construction is illustrated in Figure \ref{fig:Gd}.

\begin{figure}[ht]
\vspace{3mm}
\centering
\begin{tikzpicture}[scale=0.4,rectnode/.style={shape=circle,draw=black,minimum height=25mm},roundnode/.style={circle, draw=black!60, fill=blue!5, very thick, minimum size=6mm}]
    \newcommand\XA{3}
    \newcommand\YA{2}
    
    \node[roundnode,draw=black,minimum height=12mm] (a) at (-4.5*\XA,0) {$A$};    
    \node[roundnode,draw=black,minimum height=12mm] (P1) at (-3*\XA,0) {$P_1$}; 
    \node[roundnode,draw=black,minimum height=12mm] (P2) at (-1.75*\XA,0) {$P_2$}; 
    \node[roundnode,draw=black,minimum height=12mm] (Pd) at (1.75*\XA,0) {\small $P_{d-1}$}; 
    \node[roundnode,draw=black,minimum height=12mm] (Pdd) at (3*\XA,0) {$P_d$}; 
    \node[roundnode,draw=black,minimum height=12mm] (b) at (4.5*\XA,0) {$B$};    

    \path[] 
     (a) edge (P1)
     (P1) edge (P2)
     (Pd) edge (Pdd)
     (Pdd) edge (b);
    \draw[very thick,decorate,rotate=180,decoration={brace,amplitude=1mm}] (-3.5*\XA,2.0) -- (3.5*\XA,2.0);
    \node[draw=none,fill=none,] at (0,-3.2) {$d$ nodes};
    \node[draw=none,fill=none,] at (0,0) {......................};
\end{tikzpicture}
\caption{The graph $\Gd$.}\label{fig:Gd}
\end{figure}

Let $X$, $Y$ be two finite sets, and let $f$ be a function from $X\times Y$ to $\{0,1\}$. Consider the following computation problem over the network $\Gd$: node $A$ receives an input $x\in X$ and node~$B$ receives an input $y\in Y$. The goal is for Alice and Bob to compute the value $f(x,y)$. Note that the intermediate nodes $P_1,\ldots,P_d$ do not receive any input. 
The following theorem relates the round complexity of this problem to the two-party communication complexity of $f$. 
\begin{figure}[ht!]
\centering
\begin{tikzpicture}[scale=0.5,rectnode/.style={shape=rectangle,draw=black,minimum height=10mm},rectnode2/.style={shape=rectangle,draw=black,minimum height=8mm},roundnode/.style={circle, draw=green!60, fill=green!5, very thick, minimum size=7mm}]
    \newcommand\XA{3}
    \newcommand\YA{4}
        
    \node[rectnode] (a2) at (0*\XA,1*\YA) {};
    \node[rectnode] (a3) at (0*\XA,2*\YA) {};
    \node[rectnode] (a4) at (0*\XA,3*\YA) {};
    
    \node[rectnode] (b1) at (1*\XA,0*\YA) {};
    \node[rectnode] (b2) at (1*\XA,1*\YA) {};
    \node[rectnode] (b3) at (1*\XA,2*\YA) {};
    
    \node[rectnode] (c2) at (2*\XA,1*\YA) {};
    \node[rectnode] (c3) at (2*\XA,2*\YA) {};
    \node[rectnode] (c4) at (2*\XA,3*\YA) {};
    
    \node[rectnode] (d1) at (3*\XA,0*\YA) {};
    \node[rectnode] (d2) at (3*\XA,1*\YA) {};
    \node[rectnode] (d3) at (3*\XA,2*\YA) {};
    
    \node[rectnode] (e2) at (4*\XA,1*\YA) {};
    \node[rectnode] (e3) at (4*\XA,2*\YA) {};
    \node[rectnode] (e4) at (4*\XA,3*\YA) {};
    
    \node[rectnode] (f1) at (5*\XA,0*\YA) {};
    \node[rectnode] (f2) at (5*\XA,1*\YA) {};
    \node[rectnode] (f3) at (5*\XA,2*\YA) {};
    
    \node[rectnode] (g2) at (6*\XA,1*\YA) {};
    \node[rectnode] (g3) at (6*\XA,2*\YA) {};
    \node[rectnode] (g4) at (6*\XA,3*\YA) {};
    
    \node[rectnode] (h1) at (7*\XA,0*\YA) {};
    \node[rectnode] (h2) at (7*\XA,1*\YA) {};
    \node[rectnode] (h3) at (7*\XA,2*\YA) {};
    
    \node[rectnode2] (i1) at (8*\XA,0*\YA) {};
    \node[rectnode2] (i2) at (8*\XA,1*\YA) {};
    \node[rectnode2] (i3) at (8*\XA,2*\YA) {};
    \node[rectnode2] (i4) at (8*\XA,3*\YA) {};
    
     \path[]
          
     (a2) edge (b2)
     (a3) edge[blue] (b3)
     (a4) edge (c4)
     (a2) edge (b1)
     (a3) edge (b2)
     (a4) edge[red] (b3)
     
     (b1) edge (d1)
     (b2) edge (c2)
     (b3) edge[blue] (c3)
     (b1) edge (c2)
     (b2) edge (c3)
     (b3) edge[red] (c4)
     
     (c2) edge (d2)
     (c3) edge[blue] (d3)
     (c4) edge (e4)
     (c2) edge (d1)
     (c3) edge (d2)
     (c4) edge[red] (d3)
     
     (d1) edge (f1)
     (d2) edge (e2)
     (d3) edge[blue] (e3)
     (d1) edge (e2)
     (d2) edge (e3)
     (d3) edge[red] (e4)
     
     (e2) edge (f2)
     (e3) edge[blue] (f3)
     (e4) edge (g4)
     (e2) edge (f1)
     (e3) edge (f2)
     (e4) edge[red] (f3)
     
     (f1) edge (h1)
     (f2) edge (g2)
     (f3) edge[blue] (g3)
     (f1) edge (g2)
     (f2) edge (g3)
     (f3) edge[red] (g4)
     
     (g2) edge (h2)
     (g3) edge[blue] (h3)
     (g4) edge (i4)
     (g2) edge (h1)
     (g3) edge (h2)
     (g4) edge[red] (h3)
     
     (h1) edge (i1)
     (h2) edge (i2)
     (h3) edge[blue] (i3)
     (h1) edge (i2)
     (h2) edge (i3)
     (h3) edge[red] (i4);
     
     \draw (-1,0*\YA+0) -- (-.25+1*\XA,0*\YA+0);

     \draw (-1,1*\YA+0) -- (-.25,1*\YA+0);
     \draw (-1,1*\YA-0.35) -- (-.25,1*\YA-0.35);

     \draw[blue] (-1,2*\YA+0) -- (-.25,2*\YA+0);
     \draw (-1,2*\YA-0.35) -- (-.25,2*\YA-0.35);
     
     \draw (-1,3*\YA+0) -- (-.25,3*\YA+0);
     \draw[red] (-1,3*\YA-0.35) -- (-.25,3*\YA-0.35);
     
     \draw (8*\XA+0.25,0*\YA+0) -- (8*\XA+1,0*\YA+0);

     \draw (8*\XA+0.25,1*\YA+0) -- (8*\XA+1,1*\YA+0);
     \draw (8*\XA+0.25,1*\YA-0.35) -- (8*\XA+1,1*\YA-0.35);

     \draw[blue] (8*\XA+0.25,2*\YA+0) -- (8*\XA+1,2*\YA+0);
     \draw (8*\XA+0.25,2*\YA-0.35) -- (8*\XA+1,2*\YA-0.35);
     
     \draw (8*\XA+0.25,3*\YA+0) -- (8*\XA+1,3*\YA+0);
     \draw[red] (8*\XA+0.25,3*\YA-0.35) -- (8*\XA+1,3*\YA-0.35);
     
     
     \node[draw=none,fill=none] at (1.1,3*\YA+2.3) {$t\!=\!1$};
     \node[draw=none,fill=none] at (1.1+\XA,3*\YA+2.3) {$t\!=\!2$};
     \node[draw=none,fill=none] at (1.1+2*\XA,3*\YA+2.3) {$t\!=\!3$};
     \node[draw=none,fill=none] at (1.1+3*\XA,3*\YA+2.3) {$t\!=\!4$};
     \node[draw=none,fill=none] at (1.1+4*\XA,3*\YA+2.3) {$t\!=\!5$};
     \node[draw=none,fill=none] at (1.1+5*\XA,3*\YA+2.3) {$t\!=\!6$};
     \node[draw=none,fill=none] at (1.1+6*\XA,3*\YA+2.3) {$t\!=\!7$};
     \node[draw=none,fill=none] at (1.1+7*\XA,3*\YA+2.3) {$t\!=\!8$};
     \draw[<->] (-0.4,3*\YA+1.7) -- (2.2,3*\YA+1.7);
     \draw[<->] (-0.4+\XA,3*\YA+1.7) -- (2.2+\XA,3*\YA+1.7);
     \draw[<->] (-0.4+2*\XA,3*\YA+1.7) -- (2.2+2*\XA,3*\YA+1.7);
     \draw[<->] (-0.4+3*\XA,3*\YA+1.7) -- (2.2+3*\XA,3*\YA+1.7);
     \draw[<->] (-0.4+4*\XA,3*\YA+1.7) -- (2.2+4*\XA,3*\YA+1.7);
     \draw[<->] (-0.4+5*\XA,3*\YA+1.7) -- (2.2+5*\XA,3*\YA+1.7);
     \draw[<->] (-0.4+6*\XA,3*\YA+1.7) -- (2.2+6*\XA,3*\YA+1.7);
     \draw[<->] (-0.4+7*\XA,3*\YA+1.7) -- (2.2+7*\XA,3*\YA+1.7);
     
     \draw[very thick,decorate,decoration={brace,amplitude=1mm}] (-2,-1) -- (-2,1);
     \draw[very thick,decorate,decoration={brace,amplitude=1mm}] (-2,-1+\YA) -- (-2,1+\YA);
     \draw[very thick,decorate,decoration={brace,amplitude=1mm}] (-2,-1+2*\YA) -- (-2,1+2*\YA);
     \draw[very thick,decorate,decoration={brace,amplitude=1mm}] (-2,-1+3*\YA) -- (-2,1+3*\YA);
     \node[draw=none,fill=none] at (-3.8,0) {$B=P_3$};
     \node[draw=none,fill=none] at (-3,\YA) {$P_2$};
     \node[draw=none,fill=none] at (-3,2*\YA) {$P_1$};
     \node[draw=none,fill=none] at (-3.8,3*\YA) {$A=P_0$};
     
     \node[draw=none,fill=none] at (-1.5,-0.1) {\tiny $\regR_3$};
     
     \node[draw=none,fill=none] at (-1.5,-0.1+\YA) {\tiny $\regR_2$};
     \node[draw=none,fill=none] at (-1.5,-0.5+\YA) {\tiny $\regT_2$};
     
     \node[draw=none,fill=none] at (-1.5,-0.1+2*\YA) {\tiny $\regR_1$};
     \node[draw=none,fill=none] at (-1.5,-0.5+2*\YA) {\tiny $\regT_1$};
     
     \node[draw=none,fill=none] at (-1.5,-0.1+3*\YA) {\tiny $\regR_0$};
     \node[draw=none,fill=none] at (-1.5,-0.5+3*\YA) {\tiny $\regT_0$};
     
     \node[draw=none,fill=none] at (25.5,-0.1) {\tiny $\regR_3$};
     
     \node[draw=none,fill=none] at (25.5,-0.1+\YA) {\tiny $\regR_2$};
     \node[draw=none,fill=none] at (25.5,-0.5+\YA) {\tiny $\regT_2$};
     
     \node[draw=none,fill=none] at (25.5,-0.1+2*\YA) {\tiny $\regR_1$};
     \node[draw=none,fill=none] at (25.5,-0.5+2*\YA) {\tiny $\regT_1$};
     
     \node[draw=none,fill=none] at (25.5,-0.1+3*\YA) {\tiny $\regR_0$};
     \node[draw=none,fill=none] at (25.5,-0.5+3*\YA) {\tiny $\regT_0$};
\end{tikzpicture}
\caption{Representation of an $8$-round algorithm for $d=2$. Each thin line represents a quantum register. Each rectangle represents a unitary operator. The red line represents the message register $\regT_0$ (which is exchanged between $P_0$ and $P_1$ during the execution of the algorithm) and the green line represents the private register $R_1$ (which stays at node $P_1$ during the execution of the algorithm).}\label{fig:protocol}
\end{figure}

\begin{figure}[ht!]
\centering
\begin{tikzpicture}[scale=0.5,rectnode/.style={shape=rectangle,draw=black,minimum height=10mm},rectnode2/.style={shape=rectangle,draw=black,minimum height=8mm},roundnode/.style={circle, draw=green!60, fill=green!5, very thick, minimum size=7mm}]
    \newcommand\XA{3}
    \newcommand\YA{4}
        
    \node[rectnode] (a2) at (0*\XA,1*\YA) {};
    \node[rectnode] (a3) at (0*\XA,2*\YA) {};
    \node[rectnode] (a4) at (0*\XA,3*\YA) {};
    
    \node[rectnode] (b1) at (1*\XA,0*\YA) {};
    \node[rectnode] (b2) at (1*\XA,1*\YA) {};
    \node[rectnode] (b3) at (1*\XA,2*\YA) {};
    
    \node[rectnode] (c2) at (2*\XA,1*\YA) {};
    \node[rectnode] (c3) at (2*\XA,2*\YA) {};
    \node[rectnode] (c4) at (2*\XA,3*\YA) {};
    
    \node[rectnode] (d1) at (3*\XA,0*\YA) {};
    \node[rectnode] (d2) at (3*\XA,1*\YA) {};
    \node[rectnode] (d3) at (3*\XA,2*\YA) {};
    
    \node[rectnode] (e2) at (4*\XA,1*\YA) {};
    \node[rectnode] (e3) at (4*\XA,2*\YA) {};
    \node[rectnode] (e4) at (4*\XA,3*\YA) {};
    
    \node[rectnode] (f1) at (5*\XA,0*\YA) {};
    \node[rectnode] (f2) at (5*\XA,1*\YA) {};
    \node[rectnode] (f3) at (5*\XA,2*\YA) {};
    
    \node[rectnode] (g2) at (6*\XA,1*\YA) {};
    \node[rectnode] (g3) at (6*\XA,2*\YA) {};
    \node[rectnode] (g4) at (6*\XA,3*\YA) {};
    
    \node[rectnode] (h1) at (7*\XA,0*\YA) {};
    \node[rectnode] (h2) at (7*\XA,1*\YA) {};
    \node[rectnode] (h3) at (7*\XA,2*\YA) {};
    
    \node[rectnode2] (i1) at (8*\XA,0*\YA) {};
    \node[rectnode2] (i2) at (8*\XA,1*\YA) {};
    \node[rectnode2] (i3) at (8*\XA,2*\YA) {};
    \node[rectnode2] (i4) at (8*\XA,3*\YA) {};
    
     \path[]
          
     (a2) edge (b2)
     (a3) edge[blue,thick] (b3)
     (a4) edge (c4)
     (a2) edge (b1)
     (a3) edge (b2)
     (a4) edge (b3)
     
     (b1) edge (d1)
     (b2) edge[blue,thick] (c2)
     (b3) edge (c3)
     (b1) edge[blue,thick] (c2)
     (b2) edge[blue,thick] (c3)
     (b3) edge (c4)
     
     (c2) edge[red,thick] (d2)
     (c3) edge (d3)
     (c4) edge (e4)
     (c2) edge[red,thick] (d1)
     (c3) edge[red,thick] (d2)
     (c4) edge (d3)
     
     (d1) edge (f1)
     (d2) edge (e2)
     (d3) edge[red,thick] (e3)
     (d1) edge (e2)
     (d2) edge (e3)
     (d3) edge (e4)
     
     (e2) edge (f2)
     (e3) edge[blue,thick] (f3)
     (e4) edge (g4)
     (e2) edge (f1)
     (e3) edge (f2)
     (e4) edge (f3)
     
     (f1) edge (h1)
     (f2) edge[blue,thick] (g2)
     (f3) edge (g3)
     (f1) edge[blue,thick] (g2)
     (f2) edge[blue,thick] (g3)
     (f3) edge (g4)
     
     (g2) edge (h2)
     (g3) edge (h3)
     (g4) edge (i4)
     (g2) edge (h1)
     (g3) edge (h2)
     (g4) edge (h3)
     
    (h1) edge (i1)
     (h2) edge (i2)
     (h3) edge (i3)
     (h1) edge (i2)
     (h2) edge (i3)
     (h3) edge (i4);
     
     \draw (-1,0*\YA+0) -- (-.25+1*\XA,0*\YA+0);

     \draw (-1,1*\YA+0) -- (-.25,1*\YA+0);
     \draw (-1,1*\YA-0.35) -- (-.25,1*\YA-0.35);

     \draw (-1,2*\YA+0) -- (-.25,2*\YA+0);
     \draw (-1,2*\YA-0.35) -- (-.25,2*\YA-0.35);
     
     \draw (-1,3*\YA+0) -- (-.25,3*\YA+0);
     \draw (-1,3*\YA-0.35) -- (-.25,3*\YA-0.35);
     
     \draw (8*\XA+0.25,0*\YA+0) -- (8*\XA+1,0*\YA+0);

     \draw (8*\XA+0.25,1*\YA+0) -- (8*\XA+1,1*\YA+0);
     \draw (8*\XA+0.25,1*\YA-0.35) -- (8*\XA+1,1*\YA-0.35);

     \draw (8*\XA+0.25,2*\YA+0) -- (8*\XA+1,2*\YA+0);
     \draw (8*\XA+0.25,2*\YA-0.35) -- (8*\XA+1,2*\YA-0.35);
     
     \draw (8*\XA+0.25,3*\YA+0) -- (8*\XA+1,3*\YA+0);
     \draw (8*\XA+0.25,3*\YA-0.35) -- (8*\XA+1,3*\YA-0.35);
     
     
          \node[draw=none,fill=none] at (1.1,3*\YA+2.3) {$t\!=\!1$};
     \node[draw=none,fill=none] at (1.1+\XA,3*\YA+2.3) {$t\!=\!2$};
     \node[draw=none,fill=none] at (1.1+2*\XA,3*\YA+2.3) {$t\!=\!3$};
     \node[draw=none,fill=none] at (1.1+3*\XA,3*\YA+2.3) {$t\!=\!4$};
     \node[draw=none,fill=none] at (1.1+4*\XA,3*\YA+2.3) {$t\!=\!5$};
     \node[draw=none,fill=none] at (1.1+5*\XA,3*\YA+2.3) {$t\!=\!6$};
     \node[draw=none,fill=none] at (1.1+6*\XA,3*\YA+2.3) {$t\!=\!7$};
     \node[draw=none,fill=none] at (1.1+7*\XA,3*\YA+2.3) {$t\!=\!8$};
     \draw[<->] (-0.4,3*\YA+1.7) -- (2.2,3*\YA+1.7);
     \draw[<->] (-0.4+\XA,3*\YA+1.7) -- (2.2+\XA,3*\YA+1.7);
     \draw[<->] (-0.4+2*\XA,3*\YA+1.7) -- (2.2+2*\XA,3*\YA+1.7);
     \draw[<->] (-0.4+3*\XA,3*\YA+1.7) -- (2.2+3*\XA,3*\YA+1.7);
     \draw[<->] (-0.4+4*\XA,3*\YA+1.7) -- (2.2+4*\XA,3*\YA+1.7);
     \draw[<->] (-0.4+5*\XA,3*\YA+1.7) -- (2.2+5*\XA,3*\YA+1.7);
     \draw[<->] (-0.4+6*\XA,3*\YA+1.7) -- (2.2+6*\XA,3*\YA+1.7);
     \draw[<->] (-0.4+7*\XA,3*\YA+1.7) -- (2.2+7*\XA,3*\YA+1.7);
     \draw[very thick,decorate,decoration={brace,amplitude=1mm}] (-2,-1) -- (-2,1);
     \draw[very thick,decorate,decoration={brace,amplitude=1mm}] (-2,-1+\YA) -- (-2,1+\YA);
     \draw[very thick,decorate,decoration={brace,amplitude=1mm}] (-2,-1+2*\YA) -- (-2,1+2*\YA);
     \draw[very thick,decorate,decoration={brace,amplitude=1mm}] (-2,-1+3*\YA) -- (-2,1+3*\YA);
     \node[draw=none,fill=none] at (-3.8,0) {$B=P_3$};
     \node[draw=none,fill=none] at (-3,\YA) {$P_2$};
     \node[draw=none,fill=none] at (-3,2*\YA) {$P_1$};
     \node[draw=none,fill=none] at (-3.8,3*\YA) {$A=P_0$};
     
     \node[draw=none,fill=none] at (-1.5,-0.1) {\tiny $\regR_3$};
     
     \node[draw=none,fill=none] at (-1.5,-0.1+\YA) {\tiny $\regR_2$};
     \node[draw=none,fill=none] at (-1.5,-0.5+\YA) {\tiny $\regT_2$};
     
     \node[draw=none,fill=none] at (-1.5,-0.1+2*\YA) {\tiny $\regR_1$};
     \node[draw=none,fill=none] at (-1.5,-0.5+2*\YA) {\tiny $\regT_1$};
     
     \node[draw=none,fill=none] at (-1.5,-0.1+3*\YA) {\tiny $\regR_0$};
     \node[draw=none,fill=none] at (-1.5,-0.5+3*\YA) {\tiny $\regT_0$};
     
     \node[draw=none,fill=none] at (25.5,-0.1) {\tiny $\regR_3$};
     
     \node[draw=none,fill=none] at (25.5,-0.1+\YA) {\tiny $\regR_2$};
     \node[draw=none,fill=none] at (25.5,-0.5+\YA) {\tiny $\regT_2$};
     
     \node[draw=none,fill=none] at (25.5,-0.1+2*\YA) {\tiny $\regR_1$};
     \node[draw=none,fill=none] at (25.5,-0.5+2*\YA) {\tiny $\regT_1$};
     
     \node[draw=none,fill=none] at (25.5,-0.1+3*\YA) {\tiny $\regR_0$};
     \node[draw=none,fill=none] at (25.5,-0.5+3*\YA) {\tiny $\regT_0$};
     
     \draw[dashed,thick,draw=blue] (-0.5,2+2*\YA-0.4) -- (-0.5,-1.3) -- (1.7*\XA,-1.3) -- (1.7*\XA,2.7) -- (0,2+2*\YA-0.4) -- (-0.5,2+2*\YA-0.4);
     \draw[dashed,thick,draw=blue] (4*\XA,2+2*\YA-0.4) -- (1.8*\XA,1.1) -- (1.8*\XA,-1.3) -- (5.8*\XA,-1.3) -- (5.8*\XA,2.3) -- (4*\XA,2+2*\YA-0.4);
     \draw[dashed,thick,draw=red] (-0.5,2+2*\YA+1.2) -- (-0.5,2+2*\YA+3.2) -- (3.7*\XA,2+2*\YA+3.2) -- (3.7*\XA,9.3) -- (2*\XA,2.5) -- (-0.5,2+2*\YA+1.2);
     \draw[dashed,thick,draw=red] (3.8*\XA,2+2*\YA+3.2) -- (3.8*\XA,11.3) -- (6*\XA,2.5) -- (7.8*\XA,9.6) --(7.8*\XA,2+2*\YA+3.2) -- (3.8*\XA,2+2*\YA+3.2);
     
     \node[draw=none,fill=none,text=blue] (s0) at (1,-3) {Simulated by Bob at $s=1$};
     \node[draw=none,fill=none,text=blue] (s2) at (9,-4.5) {Simulated by Bob at $s=3$};
     \node[draw=none,fill=none,text=red] (s1) at (5,16) {Simulated by Alice at $s=2$};
     \node[draw=none,fill=none,text=red] (s3) at (17,16) {Simulated by Alice at $s=4$};
     \node[draw=none,fill=none] (ss0) at (3,-1.05){};
     \draw[<-,thick,dashed,draw=blue] (s0.north) to [out=50,in=-150] (ss0.south);
     \node[draw=none,fill=none] (ss2) at (14,-1.05){};
     \draw[<-,thick,dashed,draw=blue] (s2.north) to [out=50,in=-150] (ss2.south);
     \node[draw=none,fill=none] (ss1) at (5.5,12.9){};
     \draw[<-,thick,dashed,draw=red] (s1.south) to [out=-60,in=100] (ss1.north);  
     \node[draw=none,fill=none] (ss3) at (17.5,12.9){};
     \draw[<-,thick,dashed,draw=red] (s3.south) to [out=-60,in=100] (ss3.north);  
     
\end{tikzpicture}
\caption{Simulation by $\Pp$ of the algorithm from Figure \ref{fig:protocol}. The blue and red dashed areas represent the parts that are simulated by Bob and Alice, respectively, during the protocol. The blue plain lines represent the quantum registers that are sent by Bob to Alice during the protocol. The red plain lines represent the quantum registers that are sent by Alice to Bob during the protocol. For instance, for $s=1$, Bob will send to Alice four registers: he first simulates the computation of $P_2$ and $P_3$ until steps $t=1$ and $t=2$ of $\Aa$, respectively, and sends to Alice the corresponding message register of $P_3$ (since 3 is odd). He then simulates $P_1$ and $P_2$ until steps $t=1$ and $t=2$, respectively, and sends to Alice the message register of $P_2$ (since 2 is even). He also sends to Alice the private registers of $P_1$ and $P_2$.  }\label{fig:protocol2}
\end{figure}
\begin{theorem}\label{th:sim}
Let $d$ and $r$ be any positive integers, and $f$ be any Boolean function.
If there exists an $r$-round quantum distributed algorithm, in which each intermediate node uses at most $s$ qubits of memory, that computes function $f$ with probability $p$ over $\Gd$, then there exists a $O(r/d)$-message two-party quantum protocol computing $f$ with probability $p$ using $O(r(\band+\mem))$ qubits of communication.
\end{theorem}
\begin{proof}
Consider any $r$-round quantum algorithm $\Aa$ computing function $f$ over $\Gd$ in which each intermediate node uses $s$ qubits of internal memory.  For notational convenience, we will use $P_0$ and $P_{d+1}$ to refer to nodes $A$ and $B$, respectively. For any $i\in\{0,1,\ldots,d+1\}$, let $\regR_i$ denote the quantum register representing the memory of node $P_i$. Note that if $1\le i\le d$ then $\regR_i$ consists of at most $\mem$ qubits, which we assume (without loss of generality) are initially in the all-zero quantum state. The registers $\regR_0$ and $\regR_{d+1}$ contain the input $x$ and $y$, respectively, in addition to qubits initially in the all-zero quantum state.

For simplify notations, we will make the following assumption: at each round $t\in\{1,\ldots,r\}$ of the algorithm $\Aa$, if $t$ is odd then messages are sent only from the left to the right (i.e, from $P_i$ to $P_{i+1}$ for each $i\in\{0,\ldots,d\}$), and if $t$ is even then messages are sent only from the right to the left (i.e, from $P_i$ to $P_{i-1}$ for each $i\in\{1,\ldots,d+1\}$). This can be done without loss of generality since any algorithm can be converted into an algorithm satisfying this assumption by increasing the round complexity only by a factor 2. Without loss of generality, again, we can also assume that Algorithm~$\Aa$ proceeds as follows (see Figure \ref{fig:protocol} for an illustration for $d=2$ and $r=8$):
\begin{itemize}
\item
Each node $P_i$, for $i\in\{0,\ldots,d\}$, initially owns one additional $b$-qubit quantum register $\regT_i$ containing the all-zero quantum state.  
\item
At the first round, for each $i\in\{0,\ldots,d\}$, node $P_i$ performs a unitary operator on $(\regR_i,\regT_i)$, and then sends $\regT_i$ to $P_{i+1}$. (Note that the contents of the register $\regT_i$ can be --- and will in general be --- modified by the unitary operator.)
Node $P_{d+1}$ does nothing.
\item
At round number $t\in\{2,\ldots,r\}$ with $t$ even, for each $i\in\{1,\ldots,d+1\}$ node $P_i$, who just received register $\regT_{i-1}$ from $P_{i-1}$ at the previous round, performs a unitary operator on $(\regR_i,\regT_{i-1})$, and then sends $\regT_{i-1}$ to $P_{i-1}$. Node $P_0$ does nothing. 
\item
At round number $t\in\{3,\ldots,r\}$ with $t$ odd, for each $i\in\{0,\ldots,d\}$ node $P_i$, who just received register $\regT_{i}$ from $P_{i+1}$ at the previous round, performs a unitary operator on $(\regR_i,\regT_i)$, and then sends $\regT_i$ to $P_{i+1}$. Node $P_{d+1}$ does nothing. 
\item
After round number $r$ (the last round of communication), each node performs a unitary operator on its registers, then measures them and decides its output depending of the outcome of the measurement.
\end{itemize}
The registers $R_0,\ldots,R_{d+1}$ will be called the private registers (since they do not move during the execution of the protocol), while the registers $T_0,\ldots,T_{d}$ will be called the message registers.

We now describe a $O(r/d)$-message protocol $\Pp$ that computes $f$ in the standard model of communication complexity, where Alice and Bob receive as input $x$ and $y$, respectively. The idea is that Alice and Bob jointly simulate the computation of $\Aa$. The crucial observation is that, due to the time needed for a message from $A$ to reach $B$ in the network $\Gd$, the simulation can be done by dividing the computation into $O(r/d)$ areas, each of width $d$, and doing the simulation area per area. Each area will be simulated by one of the players, who will then send to the other player only the $O(d)$ registers that are needed to simulate the next area. The details of the simulation are as follows (see also Figure \ref{fig:protocol2} for an illustration).

 Note that Alice can create the initial state of register $\regR_0$, which depends only on $x$, and Bob can create the initial state $\regR_{d+1}$, which depends only on $y$.
For simplicity we will assume that $r$ is an even multiple of $d$, but the description can be readily adjusted for any value of $r$. 
For each $s$ from $1$ to $r/d$, Alice and Bob do the following:
\begin{itemize}
\item
If $s$ is odd then this is Bob's turn. Bob first simulates the computation of $P_i$ in $\Aa$ up to step $t=(s-1)d+i-1$, for each $i\in\{2,\ldots,d+1\}$. For each odd $i\in \{2,\ldots,d+1\}$, Bob then sends to Alice the message register of $P_i$. Bob then proceeds with the simulation and simulates the computation of $P_i$ in $\Aa$ up to step $t=(s-1)d+i$, this time for each $i\in\{1,\ldots,d\}$. For any even $i\in \{1,\ldots,d\}$, Bob sends to Alice the message register of $P_i$. Finally, Bob sends to Alice the private register of $P_i$ for each $i\in\{1,\ldots,d\}$. Note that all the communication from Bob to Alice can be concatenated into one message of $O(d(\band+\mem))$ qubits.
\item
If $s$ is even then this is Alice's turn. Alice first simulates the computation of $P_i$ in $\Aa$ up to step $t=sd-i$, for each $i\in\{0,\ldots,d-1\}$. For any odd $i\in \{0,\ldots,d-1\}$, Alice then sends to Bob the message register of $P_i$. Alice then proceeds with the simulation and simulates the computation of $P_i$ in $\Aa$ up to step $t=sd-i+1$, this time for each $i\in\{1,\ldots,d\}$. For any even $i\in \{1,\ldots,d\}$, Alice sends to Bob the message register of $P_i$. Finally, Alice sends to Bob the private register of $P_i$ for each $i\in\{1,\ldots,d\}$. Note that all the communication from Alice to Bob can be concatenated into one message of $O(d(\band+\mem))$ qubits.
\end{itemize} 
After all those steps, Alice has received all the messages needed to complete the computation done by node $A$ in Algorithm $\Aa$. She completes the computation, decides her output (as done by node $A$ in $\Aa$), and sends it to Bob.

Protocol $\Pp$ uses $O(r/d)$ messages and simulates Algorithm $\Aa$, i.e., computes the function $f$.
Its communication complexity is $O(r/d\times d(\band+\mem))=O(r(\band+\mem))$ qubits.
\end{proof}

\subsection{Proof of Theorem \ref{th:LB2}}\label{sec:proofLB2}

We are now ready to give the proof of Theorem  \ref{th:LB2}.

\begin{proof}[Proof of Theorem \ref{th:LB2}]
Let $\{G_n\}_{n\ge 1}$ be the family of graphs that realizes the $(b,k,d_1,d_2)$-reduction from disjointness to diameter computation of Theorem \ref{th:construction2}. This means that we have 
$b=\Theta(\log n)$, $k=\Theta(n)$, $d_1=4$ and $d_2=5$. We will use below the same notations as in Definition \ref{def:sim}. 

Let $d$ be any positive integer such that $d\le n$. Assume that there exists an $r$-round quantum algorithm~$\Aa$ that decides if the diameter of a network is at most~$d+d_1$ or at least $d+d_2$. For any $k\ge 0$, we will use this algorithm to construct an $r$-round algorithm that computes the disjointness function $\DISJ{k}$ in the model described in Section \ref{sec:sim} with $d+2$ players $A,P_1,\ldots,P_d,B$ and bandwidth $\band=\Theta(b\log n)$. Let $x\in \{0,1\}^{k}$ and $y\in \{0,1\}^{k}$ be the inputs of nodes~$A$ and $B$, respectively.

Consider the bipartite graph $G_n(x,y)$, which has $b$ edges between $U_n$ and $V_n$. We modify this graph to obtain a new graph $G'_n(x,y)$ by replacing each of these $b$ edges by a path of length $d+1$ corresponding to $d$ new nodes (see Figure \ref{fig:proofLB2} for an illustration). Note that $G'_n(x,y)$ contains $n'=n+bd$ nodes, i.e., the quantity $n$ does not correspond anymore to the size of the network considered. This point can nevertheless be ignored since $b=\Theta(\log n)$ in our construction, which implies $n'=\Theta(n\log n)$, and since we ignore logarithmic terms in our lower bound.
  
  The protocol works as follows. First note that the players can jointly construct the graph $G'_n(x,y)$: $A$ can construct the subgraph induced by $U_n$, $B$ can construct the subgraph induced by $V_n$, and the intermediate nodes can naturally be partitioned in $d$ vertical layers constructed by $P_1,\ldots,P_d$ (see Figure~\ref{fig:proofLB2}). Nodes $A,P_1,\ldots,P_d,B$ will thus simply simulate the computation of $\Aa$ on $G'_n(x,y)$ in $r$ rounds.  At the end of simulation $A$ and $B$ obtain the result of the computation by $\Aa$ on $G'_n(x,y)$. From our assumption on $\Aa$, this enables them to decide whether the diameter of the network is at most~$d+d_1$ or at least $d+d_2$ in $r$ rounds and with at most $\mem$ qubits of memory per node, and thus to compute $\DISJ{k}(x,y)$ with the same complexity. Theorem~\ref{th:sim} therefore implies the existence of a two-party quantum protocol for $\DISJ{k}$ with $O(r/d)$ rounds and total quantum communication complexity $O(r(b\log n+\mem))$. From Theorem \ref{th:BGK+15}, we conclude that
\[
r(b\log n+\mem)=\tilde\Omega\left(\frac{k}{r/d}+r/d\right),
\]
which implies $r=\tilde \Omega(\sqrt{kd/(b+\mem)})$. The statement of Theorem \ref{th:LB2} directly follows, since $b=\Theta(\log n)$ and $k=\Theta(n)$.
\end{proof}
\begin{figure}
\centering
\begin{tikzpicture}[scale=0.3,rectnode/.style={shape=circle,draw=black,minimum height=10mm},roundnode/.style={circle, draw=black!60, fill=blue!5, very thick, minimum size=5mm}]
    \newcommand\XA{5}
    \newcommand\YA{4}
    
    \node[roundnode,draw=black] (l0) at (-3*\XA,8) {};    
    \node[roundnode,draw=black] (l1) at (-3*\XA,5.5) {}; 
    \node[roundnode,draw=black] (lkk) at (-2*\XA,5) {};
    \node[roundnode,draw=black] (r0) at (3*\XA,8) {};    
    \node[roundnode,draw=black] (r1) at (3*\XA,5.5) {}; 
    \node[roundnode,draw=black] (rkk) at (2*\XA,5) {};
    
    \node[circle,draw=black] (s4) at (0.75*\XA,5) {};
    \node[circle,draw=black] (s3) at (0.3*\XA,5) {};
    \node[circle,draw=black] (s2) at (-0.3*\XA,5) {};
    \node[circle,draw=black] (s1) at (-0.75*\XA,5) {};
    
    \node[circle,draw=black] (t4) at (0.75*\XA,18) {};
    \node[circle,draw=black] (t3) at (0.3*\XA,18) {};
    \node[circle,draw=black] (t2) at (-0.3*\XA,18) {};
    \node[circle,draw=black] (t1) at (-0.75*\XA,18) {};
    
    \node[circle,draw=black] (u4) at (0.75*\XA,14) {};
    \node[circle,draw=black] (u3) at (0.3*\XA,14) {};
    \node[circle,draw=black] (u2) at (-0.3*\XA,14) {};
    \node[circle,draw=black] (u1) at (-0.75*\XA,14) {};
    
    \node[circle,draw=black] (v4) at (0.75*\XA,10) {};
    \node[circle,draw=black] (v3) at (0.3*\XA,10) {};
    \node[circle,draw=black] (v2) at (-0.3*\XA,10) {};
    \node[circle,draw=black] (v1) at (-0.75*\XA,10) {};
    
    \node[circle,draw=black] (w4) at (0.75*\XA,16) {};
    \node[circle,draw=black] (w3) at (0.3*\XA,16) {};
    \node[circle,draw=black] (w2) at (-0.3*\XA,16) {};
    \node[circle,draw=black] (w1) at (-0.75*\XA,16) {};
    
    \node[circle,draw=black] (y4) at (0.75*\XA,8) {};
    \node[circle,draw=black] (y3) at (0.3*\XA,8) {};
    \node[circle,draw=black] (y2) at (-0.3*\XA,8) {};
    \node[circle,draw=black] (y1) at (-0.75*\XA,8) {};
    
    \node[roundnode,draw=black] (f0) at (-2.5*\XA,14) {};
    \node[roundnode,draw=black] (fm) at (-2.5*\XA,10) {};
    \node[roundnode,draw=black] (t0) at (-1.5*\XA,14) {};
    \node[roundnode,draw=black] (tm) at (-1.5*\XA,10) {};
    \node[roundnode,draw=black] (ff0) at (1.5*\XA,14) {};
    \node[roundnode,draw=black] (ffm) at (1.5*\XA,10) {};
    \node[roundnode,draw=black] (tt0) at (2.5*\XA,14) {};
    \node[roundnode,draw=black] (ttm) at (2.5*\XA,10) {};
    
    \node[roundnode,draw=black] (a) at (-2*\XA,18) {}; 
    \node[roundnode,draw=black] (b) at (2*\XA,18) {};
  
    \draw (f0.east) to [out=20,in=180] (w1.west);
    \draw (w1) -- (w2);
    \draw (w3) -- (w4);
    \draw (y1) -- (y2);
    \draw (y3) -- (y4);    
    \draw (fm.east) to [out=-20,in=-180] (y1.west);
    \draw (tt0.west) to [out=160,in=0] (w4.east);
    \draw (ttm.west) to [out=-140,in=0] (y4.east);
    \draw (t0) -- (u1);
    \draw (u2) -- (u1);
    \draw (u3) -- (u4);
    \draw (u4) -- (ff0);
    \draw (tm) -- (v1);
    \draw (v2) -- (v1);
    \draw (v3) -- (v4);
    \draw (v4) -- (ffm);
    
    \draw (a) -- (t1);
    \draw (t1) -- (t2);
    \draw (t3) -- (t4);
    \draw (t4) -- (b);
    \draw (a) -- (f0);
    \draw (a) -- (t0);
    \draw (a) -- (fm);
    \draw (a) -- (tm);
    \draw (a) -- (lkk);

    \draw (b) -- (ff0);
    \draw (b) -- (tt0);
    \draw (b) -- (ffm);
    \draw (b) -- (ttm);
    \draw (b) -- (rkk);
        
    \draw (lkk) -- (s1);
    \draw (s1) -- (s2);
    \draw (s3) -- (s4);
    \draw (s4) -- (rkk);

    \node[draw=none,fill=none] at (0,5) {...};
    \node[draw=none,fill=none] at (0,18) {...};
    \node[draw=none,fill=none] at (0,10) {...};
    \node[draw=none,fill=none] at (0,8) {...};
    \node[draw=none,fill=none] at (0,14) {...};
    \node[draw=none,fill=none] at (0,16) {...};
    
    \draw (l0) -- (f0);
    \draw (l0) -- (fm);
    \draw (l1) -- (tm);
    \draw (l1.west) to [out=110,in=-150] (f0.west);
    \draw (r0.north) to [out=100,in=-10] (ff0);
    \draw (r0) -- (ffm);
    \draw (r1) -- (ttm);
    \draw (r1.west)  to [out=120,in=-60] (ff0.south);
    \draw[dashed, thick] (-18,-2) -- (-5.5,-2) -- (-5.5,20) -- (-18,20) -- (-18,-2);
    \draw[dashed, thick] (18,-2) -- (5.5,-2) -- (5.5,20) -- (18,20) -- (18,-2);
    \draw[dashed, thick] (-4.6,-2) -- (-2.9,-2) -- (-2.9,20) -- (-4.6,20) -- (-4.6,-2);
    \draw[dashed, thick] (-2.3,-2) -- (-0.7,-2) -- (-0.7,20) -- (-2.3,20) -- (-2.3,-2); 
    \draw[dashed, thick] (4.6,-2) -- (2.9,-2) -- (2.9,20) -- (4.6,20) -- (4.6,-2);
    \draw[dashed, thick] (2.3,-2) -- (0.7,-2) -- (0.7,20) -- (2.3,20) -- (2.3,-2);     
    \node[draw=none,fill=none] at (-11,21) {\Large $A$};
    \node[draw=none,fill=none] at (-4,21) { $P_1$};
    \node[draw=none,fill=none] at (-1.9,21) { $P_2$};
    \node[draw=none,fill=none] at (1.6,21) { $P_{d-1}$};
    \node[draw=none,fill=none] at (3.8,21) { $P_d$};
    \node[draw=none,fill=none] at (11,21) {\Large $B$};
    
    \draw[red] (-11.5,10.2) ellipse (5.6cm and 9.4cm);
    \node[draw=none,fill=none,red] at (-11.5,-.3) {\Large $U_n$};
    \draw[red] (11.5,10.2) ellipse (5.6cm and 9.4cm);
    \node[draw=none,fill=none,red] at (11.5,-.3) {\Large $V_n$};    
\end{tikzpicture}
\caption{Example of graph $G'_n(x,y)$. Here $b=6$: there are $6$ edges between $U_n$ and $V_n$ in $G_n(x,y)$, and each of these edge is replaced by a path of length $d+1$ in $G'_n(x,y)$.}\label{fig:proofLB2}
\end{figure}

\section*{Acknowledgments}
The authors are grateful to Rahul Jain, Ashwin Nayak and  Dave Touchette for helpful discussions, and to the anonymous referees for valuable comments and suggestions. 
FLG was partially supported by the JSPS KAKENHI grants No.~16H01705 and No.~16H05853.
FM was partially supported by the ERA-NET Cofund in Quantum Technologies 
project QuantAlgo and the French ANR Blanc project RDAM.


\begin{thebibliography}{HPRW14}

\bibitem[AA05]{Aaronson+05}
Scott Aaronson and Andris Ambainis.
\newblock Quantum search of spatial regions.
\newblock {\em Theory of Computing}, 1(1):47--79, 2005.

\bibitem[ACHK16]{Abboud+DISC16}
Amir Abboud, Keren Censor-Hillel, and Seri Khoury.
\newblock Near-linear lower bounds for distributed distance computations, even
  in sparse networks.
\newblock In {\em Proceedings of the 30th International Symposium on
  Distributed Computing (DISC 2016)}, pages 29--42, 2016.

\bibitem[BCW98]{Buhrman+STOC98}
Harry Buhrman, Richard Cleve, and Avi Wigderson.
\newblock Quantum vs. classical communication and computation.
\newblock In {\em Proceedings of the 30th {ACM} Symposium on the Theory of
  Computing (STOC 1998)}, pages 63--68, 1998.

\bibitem[Ben89]{Bennett+SICOMP89}
Charles~H. Bennett.
\newblock Time/space trade-offs for reversible computation.
\newblock {\em SIAM Journal on Computing}, 18(4):766--776, 1989.

\bibitem[BGK{\etalchar{+}}15]{BGK+15}
Mark Braverman, Ankit Garg, Young Kun{-}Ko, Jieming Mao, and Dave Touchette.
\newblock Near-optimal bounds on bounded-round quantum communication complexity
  of disjointness.
\newblock In {\em Proceedings of the 56th Annual IEEE Symposium on Foundations
  of Computer Science (FOCS 2015)}, pages 773--791, 2015.

\bibitem[BHT98]{Brassard+ICALP98}
Gilles Brassard, Peter H{\o}yer, and Alain Tapp.
\newblock Quantum counting.
\newblock In {\em Proceedings of the 25th International Colloquium on Automata,
  Languages and Programming (ICALP 1998)}, pages 820--831, 1998.

\bibitem[BK17]{Bringmann+DISC17}
Karl Bringmann and Sebastian Krinninger.
\newblock A note on hardness of diameter approximation.
\newblock In {\em Proceedings of the 31st International Symposium on
  Distributed Computing (DISC 2017)}, pages 44:1--44:3, 2017.

\bibitem[BT08]{Broadbent+08}
Anne Broadbent and Alain Tapp.
\newblock Can quantum mechanics help distributed computing?
\newblock {\em SIGACT News}, 39(3):67--76, 2008.

\bibitem[DHHM06]{Durr+SICOMP06}
Christoph D{\"u}rr, Mark Heiligman, Peter H{\o}yer, and Mehdi Mhalla.
\newblock Quantum query complexity of some graph problems.
\newblock {\em SIAM Journal on Computing}, 35(6):1310--1328, 2006.

\bibitem[DP08]{Denchev+08}
Vasil~S. Denchev and Gopal Pandurangan.
\newblock Distributed quantum computing: a new frontier in distributed systems
  or science fiction?
\newblock {\em SIGACT News}, 39(3):77--95, 2008.

\bibitem[EKNP14]{Elkin+PODC14}
Michael Elkin, Hartmut Klauck, Danupon Nanongkai, and Gopal Pandurangan.
\newblock Can quantum communication speed up distributed computation?
\newblock In {\em Proceedings of the 2014 {ACM} Symposium on Principles of
  Distributed Computing (PODC 2014)}, pages 166--175, 2014.

\bibitem[FHW12]{Frischknecht+SODA12}
Silvio Frischknecht, Stephan Holzer, and Roger Wattenhofer.
\newblock Networks cannot compute their diameter in sublinear time.
\newblock In {\em Proceedings of the 23rd Annual {ACM-SIAM} Symposium on
  Discrete Algorithms (SODA 2012)}, pages 1150--1162, 2012.

\bibitem[GKM09]{Gavoille+DISC09}
Cyril Gavoille, Adrian Kosowski, and Marcin Markiewicz.
\newblock What can be observed locally?
\newblock In {\em Proceedings of the 23rd International Symposium on
  Distributed Computing (DISC 2009)}, pages 243--257, 2009.

\bibitem[HdW02]{Hoyer+STACS02}
Peter H{\o}yer and Ronald de~Wolf.
\newblock Improved quantum communication complexity bounds for disjointness and
  equality.
\newblock In {\em Proceedings of the 19th Annual Symposium on Theoretical
  Aspects of Computer Science (STACS 2002)}, pages 299--310, 2002.

\bibitem[HPRW14]{Holzer+DISC14}
Stephan Holzer, David Peleg, Liam Roditty, and Roger Wattenhofer.
\newblock Distributed 3/2-approximation of the diameter.
\newblock In {\em Proceedings of the 28th International Symposium on
  Distributed Computing (DISC 2014)}, pages 562--564, 2014.

\bibitem[HW12]{Holzer+PODC12}
Stephan Holzer and Roger Wattenhofer.
\newblock Optimal distributed all pairs shortest paths and applications.
\newblock In {\em Proceedings of the 2012 ACM Symposium on Principles of
  Distributed Computing (PODC 2012)}, pages 355--364, 2012.

\bibitem[JRS03]{Jain+FOCS03}
Rahul Jain, Jaikumar Radhakrishnan, and Pranab Sen.
\newblock A lower bound for the bounded round quantum communication complexity
  of set disjointness.
\newblock In {\em Proceedings of the 44th Annual IEEE Symposium on Foundations
  of Computer Science (FOCS 2003)}, pages 220--229, 2003.

\bibitem[KS92]{Kalyanasundaram+92}
Bala Kalyanasundaram and Georg Schnitger.
\newblock The probabilistic communication complexity of set intersection.
\newblock {\em {SIAM} Journal on Discrete Mathematics}, 5(4):545--557, 1992.

\bibitem[{LeG}09]{LeGall09}
Fran{\c{c}}ois {Le Gall}.
\newblock Exponential separation of quantum and classical online space
  complexity.
\newblock {\em Theory of Computing Systems}, 45(2):188--202, 2009.

\bibitem[LP13]{Lenzen+PODC13}
Christoph Lenzen and David Peleg.
\newblock Efficient distributed source detection with limited bandwidth.
\newblock In {\em Proceedings of the 2013 {ACM} Symposium on Principles of
  Distributed Computing (PODC 2013)}, pages 375--382, 2013.

\bibitem[MNRS11]{Magniez+SICOMP11}
Fr{\'e}d{\'e}ric Magniez, Ashwin Nayak, J{\'e}r{\'e}mie Roland, and Miklos
  Santha.
\newblock Search via quantum walk.
\newblock {\em SIAM Journal on Computing}, 40(1):142--164, 2011.

\bibitem[NC11]{Nielsen+00}
Michael~A. Nielsen and Isaac~L. Chuang.
\newblock {\em Quantum Computation and Quantum Information}.
\newblock Cambridge University Press, 2011.

\bibitem[PRT12]{Peleg+ICALP12}
David Peleg, Liam Roditty, and Elad Tal.
\newblock Distributed algorithms for network diameter and girth.
\newblock In {\em Proceedings of the 39th International Colloquium on Automata,
  Languages, and Programming (ICALP 2012)}, pages 660--672, 2012.

\bibitem[Raz92]{Razborov92}
Alexander~A. Razborov.
\newblock On the distributional complexity of disjointness.
\newblock {\em Theoretical Computer Science}, 106(2):385--390, 1992.

\bibitem[Raz03]{Razborov03}
Alexander~A. Razborov.
\newblock Quantum communication complexity of symmetric predicates.
\newblock {\em Izvestiya of the Russian Academy of Science: Mathematics},
  67(1):159--176, 2003.

\bibitem[TKM12]{Tani+12}
Seiichiro Tani, Hirotada Kobayashi, and Keiji Matsumoto.
\newblock Exact quantum algorithms for the leader election problem.
\newblock {\em {ACM} Transactions on Computation Theory}, 4(1):1:1--1:24, 2012.

\bibitem[{deW}02]{deWolf02}
Ronald {de Wolf}.
\newblock Quantum communication and complexity.
\newblock {\em Theoretical Computer Science}, 287(1):337--353, 2002.

\bibitem[Yao79]{YaoSTOC79}
Andrew~C. Yao.
\newblock Some complexity questions related to distributive computing
  (preliminary report).
\newblock In {\em Proceedings of the 11th Annual ACM Symposium on Theory of
  Computing (STOC 1979)}, pages 209--213, 1979.

\bibitem[Yao93]{YaoFOCS93}
Andrew~C. Yao.
\newblock Quantum circuit complexity.
\newblock In {\em Proceedings of the 34th Annual IEEE Symposium on Foundations
  of Computer Science (FOCS 1993)}, pages 352--361, 1993.

\end{thebibliography}
\newcommand{\etalchar}[1]{$^{#1}$}

\end{document}